\documentclass[aps,pra,10pt,groupedaddress,superscriptaddress,tightenlines,twocolumn]{revtex4-2}

\usepackage{config}

\newif\ifshowcomments
\showcommentsfalse          

\ifshowcomments
  \newcommand{\hft}[1]{\textcolor{blue}{HFT: #1}}
  \newcommand{\bs}[1]{\textcolor{red}{BS: #1}}
  \newcommand{\ea}[1]{\textcolor{orange}{EA: #1}}
\else
  \newcommand{\hft}[1]{}
  \newcommand{\bs}[1]{}
  \newcommand{\ea}[1]{}
\fi

\begin{document}

\title{Efficient classical algorithm for estimating linear statistics of Boson Sampling}

\author{Benoit Seron}
\affiliation{International Iberian Nanotechnology Laboratory (INL),
Av. Mestre Jos\'e Veiga, 4715-330 Braga, Portugal}
\affiliation{Physikalisches Institut, Albert-Ludwigs-Universit\"at Freiburg,
Hermann-Herder-Stra{\ss}e 3, D-79104 Freiburg, Germany}
\affiliation{EUCOR Centre for Quantum Science and Quantum Computing,
Albert-Ludwigs-Universit\"at Freiburg, Hermann-Herder-Stra{\ss}e 3, D-79104 Freiburg, Germany}
\author{Hugo Thomas}
\affiliation{Laboratoire d’Informatique de Paris 6, CNRS, Sorbonne Université, Paris, France}
\affiliation{Quandela, 7 rue Léonard de Vinci, Massy, France}
\affiliation{DIENS, \'Ecole Normale Supérieure, PSL University, CNRS, INRIA, Paris, France}
\author{Eduardo Araujo}
\affiliation{International Iberian Nanotechnology Laboratory (INL),
Av. Mestre Jos\'e Veiga, 4715-330 Braga, Portugal}
\affiliation{Department of Engineering, University of Minho, R. da Universidade, 4710-057 Braga, Portugal}
\author{Alex Arkhipov}
\affiliation{Independent researcher}
\author{Changhun Oh}
\affiliation{Department of Physics, Korea Advanced Institute of Science and Technology, Daejeon 34141, Korea}
\author{Leonardo Novo}
\affiliation{International Iberian Nanotechnology Laboratory (INL),
Av. Mestre Jos\'e Veiga, 4715-330 Braga, Portugal}

\date{\today}

\begin{abstract}
    Boson Sampling is a prominent candidate for the demonstration of quantum computational advantage but it remains unclear whether a large-scale boson sampler can find useful computational applications. The challenge is that to use a boson sampler to estimate, for example, a physical observable, it is necessary to coarse grain the outcome distribution, owing to the exponential size of the outcome space and the anti-concentration properties of the Boson Sampling distribution. 
    In this work, we analyse the complexity of a specific type of coarse-graining of Boson Sampling distributions based on linear functions of the output photon occupation numbers, which we refer to as linear statistics. We present an efficient classical algorithm for approximating linear statistics of boson samplers within additive error for different kinds of input states, such as Fock states and squeezed states. This allows us to unify in the same framework recent results on efficient quantum-inspired classical algorithms for simulating molecular vibronic spectra, or efficient classical approximations of coarse-grained distributions based on detector binning. Additionally, we show how our algorithm can be used to classically evaluate a proposed one-way function based on Boson Sampling, while other proposals for cryptographic applications escape our classical simulation techniques. We leave open the question of classical simulability of non-linear statistics and connect it to the problem of computing transition amplitudes of linear-optical circuits with one layer of interactions.

\end{abstract}

\maketitle

\section{Introduction}

Photonic architectures are promising platforms for demonstrating quantum
advantage, with Boson Sampling \cite{aaronson_computational_2011} and Gaussian
Boson Sampling \cite{hamilton_gaussian_2017} as leading candidates \cite{hangleiter2023computational}.
Implementations of Boson Sampling have steadily improved in scale and
performance
\cite{broome_photonic_2013,tillmann_experimental_2013,loredo_boson_2017,he_timebinencoded_2017,wang_highefficiency_2017,wang_boson_2019a,zhong_quantum_2020,hoch_reconfigurable_2022,madsen_quantum_2022,maring_versatile_2024,young_atomic_2024,carosini_programmable_2024,liu_robust_2025}.
In parallel, significant theoretical progress has been made in understanding the
computational complexity of these systems, including refined hardness results
with minimal conjectures and extensions to experimentally relevant regimes
\cite{bouland_complexitytheoretic_2023,bouland_exponential_2025,ehrenberg_transition_2025,kolarovszki_general_2026,mhiri_boson_2026}.
Beyond complexity-theoretic improvements and experimental milestones, a growing
body of work has explored whether such models may also support more structured
computational tasks \cite{salavrakos_photonnative_2025}. Most prominent
examples of such tasks involve molecular spectra
\cite{huh_vibronic_2015,huh_yung_vibronic_2017}, graph problems
\cite{bradler2018gBS_perfect_matchings, arrazola_dense_2018, mezher_solving_2023}, variational or machine learning approaches \cite{agresti_demonstration_2025,yin_experimental_2025,monbroussou_classical_2026} or cryptographic applications
\cite{nikolopoulos_decision_2016,nikolopoulos_cryptographic_2019,shi_quantum_2022,singh_proofofwork_2025}.
Experimental demonstrations of solving these problems
\cite{deng_solving_2023,wang_experimental_2023,wang_efficient_2020,zhu_largescale_2024,eickmann_bridging_2026}
were motivated by their applications in a wide range of fields, with the aim of
achieving quantum computational advantage without requiring a fault-tolerant
quantum computer.

Although an outcome probability of Boson Sampling is $\#\mathsf{P}$-hard to
compute exactly in general because of its link to a matrix permanent or Hafnian (for Gaussian Boson Sampling)
\cite{aaronson_computational_2011, hamilton_gaussian_2017}, many of the potential applications of Boson Sampling require the estimation of specific  physical observables, which can only be efficiently approximated by the experiment up to an additive error. 
Since individual output probabilities are typically exponentially small, experimentally accessible quantities generally involve aggregating many outcomes.
In this context, the known computational hardness arguments do not apply and, in fact, efficient classical algorithms for estimating individual outcome 
probabilities, and more generally expectation values of some observables within an
inverse-polynomial additive error (in system size) exist
\cite{aaronson_generalizing_2014,lim_classical_2025,thomas_shedding_2025}.  Classical simulability results are even stronger when noise is taken into account, as efficient classical simulations of noisy Boson Sampling
become possible \cite{qi_regimes_2020,oh_classical_2021,oh_classical_2024}. As
it turns out, many of the proposed applications of Boson Sampling admit an
efficient classical simulation algorithm
\cite{oh_quantuminspired_2024,oh_quantuminspired_2024a,seron_efficient_2024}.
Nonetheless, cryptography-related applications, including one-way functions
\cite{nikolopoulos_cryptographic_2019} and related proposals for a quantum hash function
\cite{shi_quantum_2022} or quantum proof-of-work protocols \cite{singh_proofofwork_2025}, seem to escape these simulation methods.

A common ground between the proposals for using Boson Sampling for simulating molecular vibronic spectra and realizing 
cryptographic one-way functions is the need to coarse grain the output distribution. In this setting, outcomes are grouped together by a pre-defined
classical post-processing, and one is interested in the probabilities of each
grouped event. Notably, coarse-graining methods also play an important role in
validating the correct functioning of Boson Sampling devices
\cite{drummond2022simulating,seron_efficient_2024,bressanini_gaussian_2024,anguita_experimental_2025} by jointly counting the number of photons in groups of detectors.

In this work, we develop efficient classical algorithms for a specific type of coarse-graining strategy which we refer to as \emph{linear statistics} of the
boson-sampler outcomes. For a fixed weight vector $\bm \omega$ and denoting the output photon number distribution as $\s$, we use the term linear statistics to refer to the probability distribution of the values of $\bm \omega\cdot \s$.  We show that linear statistics, and related cumulative distributions, can be classically estimated efficiently to the same level of accuracy as running a
Boson Sampling experiment a polynomial number of times. While we focus on Boson Sampling with single-photon Fock input states, we argue that our
algorithm can be extended to other sampling schemes using different types of
input states or transformations, namely Gaussian Boson Sampling (GBS) \cite{hamilton_gaussian_2017} and linear interference of other types of product input states. This result allows us to unify in the same framework recent results on efficient quantum-inspired classical algorithms for simulating molecular vibronic spectra \cite{oh_quantuminspired_2024}, or efficient classical approximations of coarse-grained distributions based on detector binning \cite{seron_efficient_2024, anguita_experimental_2025}. In addition, we prove 
that the most probable bin problem defined in \cite{nikolopoulos_decision_2016} and used for the proposal of a one-way function in \cite{nikolopoulos_cryptographic_2019} can be evaluated classically
efficiently, while variants of this proposal seem to escape our simulability techniques \cite{shi_quantum_2022, singh_proofofwork_2025}.
As experimental imperfections are
inevitable, we also discuss how linear statistics can be simulated in the presence of partial photon
distinguishability and photon loss.

These results raise the question of finding other estimation problems that can be solved by efficient classical post-processing of boson
sampling outcomes but cannot be currently solved efficiently without access to quantum devices. As a first step in this direction, we propose a coarse-graining strategy based on a simple quadratic function of the outcome photon-number distribution and argue that simulating such non-linear statistics is tightly connected to open questions regarding simulability of linear optics with one layer of photon-photon interactions \cite{jabbour2025complexity, upreti_exponentiallyimproved_2026}.

\section{Background}

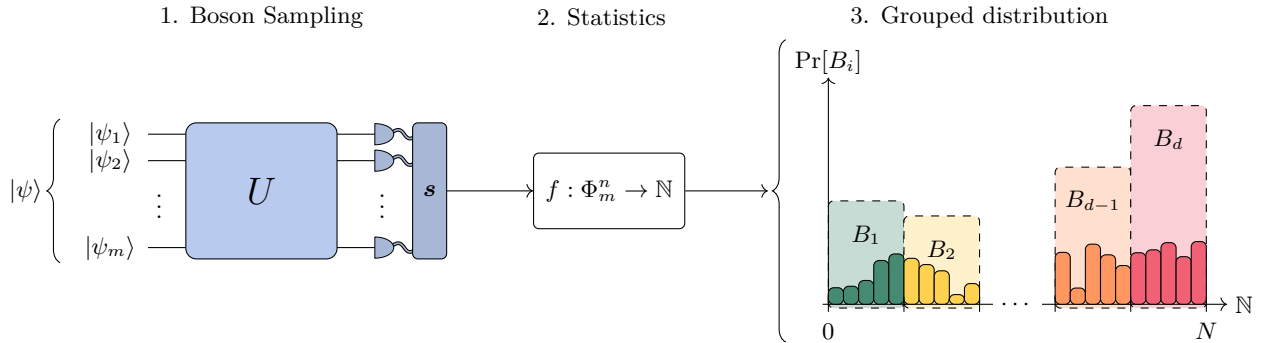
\begin{figure*}[t]
\begin{tikzpicture}[
myRectangle/.style 2 args = {
    draw, rounded corners, fill=quandelablue!40,
    inner sep=0pt, outer sep=0pt,
    fit=(#1) (#2)}, 
myPermutationRectangle/.style 2 args = {
    draw, rounded corners, fill=quandelablue!40,
    inner sep=0pt, outer sep=0pt,
    fit=(#1) (#2)}, 
myUnitary/.style 2 args = {
    draw, rounded corners, fill=blue2!40,
    inner sep=0pt, outer sep=0pt,
    fit=(#1) (#2)}, 
myPrepUnitary/.style 2 args = {
    draw, rounded corners, fill=fancygreen!40,
    inner sep=0pt, outer sep=0pt,
    fit=(#1) (#2)}, 
myFadedRectangle/.style 2 args = {
    fill=quandelablue!40,
    draw, rounded corners, path fading=custom fade out,
    inner sep=0pt, outer sep=10pt,
    fit=(#1) (#2)},  
myFadedUnitary/.style 2 args = {
    fill=quandelared!40,
    draw, rounded corners, path fading=custom fade out,
    inner sep=0pt, outer sep=10pt,
    fit=(#1) (#2)},
myFadedPrepUnitary/.style 2 args = {
    fill=fancygreen!40,
    draw, rounded corners, path fading=custom fade out,
    inner sep=0pt, outer sep=10pt,
    fit=(#1) (#2)},
myDetectorStyle/.style = {
    line width=.3pt, black, fill = blue3!40
},
myDectectorArrowStyle/.style = {
    myDetectorStyle, rounded corners=2pt, double distance=1pt, fill = white
},
optical splitter/.pic = {
  \draw[pic actions] (-.5,.25) coordinate (#1-in-1) to [out=0,in=180] (0,0)
    (-.5,-.25) coordinate (#1-in-2) to [out=0,in=180] (0,0)
    (0,0)   to [out=0,in=180] (.5, 0.25) coordinate (#1-out-1)
    (0,0)   to [out=0,in=180] (.5, -0.25) coordinate (#1-out-2);
    },
every node/.style=inner sep = 0,
]

\coordinate (bottom left Uprep) at (-6, 2.5+.1);
\coordinate (top right Uprep) at (-4, 4.5-.1);

\coordinate (bottom left Ui) at (-3.5, 2.5+.1);
\coordinate (top right Ui) at (-1.5, 4.5-.1);

\pgfmathsetmacro{\detecYrad}{.14} 
\pgfmathsetmacro{\detecXrad}{.25} 
\pgfmathsetmacro{\yshift}{1.5} 
\pgfmathsetmacro{\ysshift}{3} 
\pgfmathsetmacro{\dist}{.35} 

\coordinate (bottom left) at (0.5,0+.1);
\coordinate (top right) at (2.5, 1.5-.1);
\coordinate (bottom left shift) at (0.5,0+\yshift+.1);
\coordinate (top right shift) at (2.5, 1.5+\yshift-.1);
\coordinate (bottom left sshift) at (0.5,0+\ysshift+.1);
\coordinate (top right sshift) at (2.5, 1.5+\ysshift-.1);

\coordinate (bottom left mshift) at (0.5,0-\yshift+.1);
\coordinate (top right mshift) at (2.5, 1.5-\yshift-.1);

\coordinate (bottom left p) at (-1,-1.5+.1);
\coordinate (top right p) at (0.25, 1.5+\ysshift-.1);

\foreach \i/\j in {0/1,0.35/2, 1.5/m}
{    
    \draw[thin] (-4, 4.25-\i) -- (-1 , 4.25-\i);
    \draw[double=blue3!40, thin] (-1 + .2 ,4.25-\i) .. controls (-.87 + .3, 4.45-\i) and (-.75 + .1, 4.05-\i) .. (-.63 + .2 ,4.25-\i); 
    \draw[myDetectorStyle] (-1, 4.25 - \i + \detecYrad) arc(90:-90:\detecXrad cm and \detecYrad cm) -- (-1,4.25 - \i + \detecYrad);
    \node[] at (-4.5, 4.25-\i) {$\ket{\psi_{\j}}$};
}

\foreach \i in {-3.85, -.9}
    \node[] at (\i, 3.4) {$\vdots$};

\node[myUnitary={bottom left Ui}{top right Ui}] {};
        
\node[] at (-2.5, 4.25-.75) {\Large $U$};

\draw [decorate,decoration={brace,amplitude=5pt,raise=-1ex,mirror}]
  (-5.3, 4.25+.2) -- (-5.3, 2.75-.2) node[midway,xshift=-1em]{$\ket{\psi}$};

\node[align=center, draw, rounded corners = 2pt, minimum height = 1cm, minimum width = 2cm] (f) at (2.1, 3.5) {$f: \Phi_m^n \to \mathbb{N}$};
\draw[->] (-.22, 3.5) -- (f);

\draw[draw=black, fill=white, rounded corners = 2pt, fill=blue3!40] (-.5, 2.6) rectangle ++ (.45, 1.8);
\node[] (s) at (-.275, 3.5) {$\bm s$};

\node[] at (-2.5, 5.8) {1. Boson Sampling};
\node[] at (2, 5.8) {2. Statistics};
\node[] at (7, 5.8) {3. Grouped distribution};

\begin{scope}[shift={(4,-3)}]

\foreach \x/\col/\lab/\y in {
    0/box1/1/1.417881,
    1/box2/2/1.218620,
    3/box3/d-1/1.864111,
    4/box4/d/2.677960131813773
}{
    \draw[dashed, fill=white, rounded corners=2pt, fill=\col!30] 
        (1.0+\x, 4.95) rectangle ++ (1, \y);

    \node[] at (1.5+\x, 4.5+\y) {$B_{\lab}$};
}

\node[] (top) at (6.5, 5) {$\mathbb{N}$};

\node[
    label={[shift={(0.1,-.6)}]$0$},
    inner sep=0pt
] (bot) at (.9, 5) {};

\node[
    label={[shift={(0,-.6)}]$N$},
    inner sep=0pt
] at (6, 5) {};

\draw[->] (bot) -- (top);

\foreach \i/\x in {
    0/0.218943,
    1/0.238756,
    2/0.315023,
    3/0.579286,
    4/0.665873
}{
    \draw[fill=white, rounded corners=2pt, fill=box1] 
        (1.0+.2*\i, 5) rectangle ++ (.2, \x);     
}

\foreach \i/\x in {
    0/0.608256,
    1/0.5253,
    2/0.442812,
    3/0.1283,
    4/0.273953
}{
    \draw[fill=white, rounded corners=2pt, fill=box2] 
        (2.0+.2*\i, 5) rectangle ++ (.2, \x);     
}

\foreach \i/\x in {
    0/0.686612,
    1/0.214843,
    2/0.795344,
    3/0.654621,
    4/0.512692
}{
    \draw[fill=white, rounded corners=2pt, fill=box3] 
        (4.0+.2*\i, 5) rectangle ++ (.2, \x);     
}

\foreach \i/\x in {
    0/0.682377,
    1/0.720514,
    2/0.815195,
    3/0.62983,
    4/0.830044
}{
    \draw[fill=white, rounded corners=2pt, fill=box4] 
        (5.0+.2*\i, 5) rectangle ++ (.2, \x);     
}

\foreach \x in {0, 1, 2, 3, 4, 5}{
    \draw[black] (1+\x, 4.9) -- (1+\x, 5.1);
}

\draw[->] (1, 5) -- (1, 8);

\node[] at (1, 8.2) {$\text{Pr}[B_i]$};

\node[fill=white] at (3.5, 5) {$\cdots$};

\end{scope}

\draw [decorate,decoration={brace,amplitude=5pt,raise=-1ex}] (4.3, 1.5) -- (4.3, 5.5);
\draw[->] (f) -- (4.2, 3.5);

\end{tikzpicture}
    \caption{Informal illustration of post-processing of Boson Sampling outcome statistics. The product input state
    ${\ket{\psi} = \otimes_{i=1}^m \ket{\psi_i}}$ evolves in the linear optical network described by $U$. The
    measurement outcome $\s$ is associated with an integer $f(\s)$, which falls within some range of values defined by the intervals $B_i$. Repeating this
    process gives a histogram of estimates of grouped probabilities $\Pr[B_1], \cdots,
    \Pr[B_d]$, which are then used for various applications. We focus on the classical simulability of the distribution induced by a linear function $f(\s)$, which we refer to as linear statistics. }
    \label{fig:linearStatistics}
\end{figure*}

In this section, we consider only input Fock states of $n$ photons evolving through an $m$-mode linear optical
network made of phase shifters and beam splitters. Nevertheless our main results  hold for more general product states of the form $\ket{\psi} = \otimes_{i=1}^m \ket{\psi_i}$, which includes also Gaussian states and certain superpositions of Fock states (see \cref{sec:extensions}). Linear-optical evolutions are described
by a unitary matrix $U \in \Ub(m)$. An orthonormal basis of the carrier Hilbert
space of states with $n$ single photons is labeled by 
\begin{equation}\label{eq:phimn}
    \Phi_m^n = \left\{(s_1, \dots, s_m) \ | \ \textstyle\sum_i s_i = n, \, s_i\geq 0 \right\},
\end{equation}
which corresponds to all the ways $n$ particles can be spread across $m$ modes,
allowing multiple particles per mode. Direct counting yields that this Hilbert
space is of dimension $|\Phi_m^n| = \binom{n+m-1}{n}$. 
Given a reference input state, in general the \emph{collision-free} input state $\ket{1,\dots,1,0\dots0}$, we denote by $\Dc_U$ the probability distribution of measurement outcomes $\s \in \Phi_m^n$ induced by photon-number measurement after the reference input state evolved in the linear-optical interferometer described by $U$.

While sampling from the distribution induced by photon number measurements at the output of a linear optical network is believed to be hard classically,
it is possible to estimate each output probability to inverse-polynomial
additive precision \cite{aaronson_computational_2011}. More precisely, the best known classical algorithm for approximating transition amplitudes of Boson Sampling allows one to estimate $\mel{{\psi}}{\hat U}{{\phi}}$ to additive precision $\varepsilon$ with probability $1-\delta$ in time $O(m^2/\varepsilon^2 \log 1/\delta)$ \cite{lim_classical_2025}, where $U \in \mathbb U(m)$ and $\hat U$ is the corresponding linear-interferometer operator acting on the Fock basis.  The estimation works for arbitrary product input states $\ket{\psi}$ and $\ket{\phi}$ as long as each state appearing in the product has an efficiently computable  decomposition in the coherent state basis (see \cite{lim_classical_2025} for details).

The notion of \emph{linear statistics}  was initially introduced for verification of Boson Sampling \cite{arkhipov_2017} and is based on the analysis of the probability distribution of integer
weighted combinations of photon occupations, which feature marginal distributions and binned-mode distributions as special cases. It can be seen as a generic way of coarse-graining the Boson Sampling distribution via a linear function of the output, which also appears in different proposals for computational applications of Boson Sampling \cite{huh_vibronic_2015, nikolopoulos_cryptographic_2019}. 

\begin{definition}[Linear statistics of Boson Sampling]\label{def:linear-statistics}
    The \emph{linear statistic} with non-negative integer-valued weight-vector ${\bm \omega =
    (\omega_1, \cdots, \omega_m) \in \Nb^m_{\geq 0}}$ is defined as the probability distribution induced by
    $\Pr_{\s \sim \Dc_U}[\bm\omega \cdot \s = k]$, for $k \in \Nb$,
    where $\bm\omega \cdot \s = \omega_1 s_1 + \cdots + \omega_ms_m$.
\end{definition}
We discuss in \cref{sec:general-weight} how \cref{def:linear-statistics} can be extended to negative- and rational-valued weight-vectors, by mapping both cases to the non-negative integer-valued case.
For a fixed number of modes $m$ and photons $n$, a given weight-vector $\bm
\omega \in \mathbb{N}^m_{\geq 0}$ thus maps measurement outcomes to a set of integers $\{0, \dots, N-1\}$,
where $N$ depends on $n, m$ and $\bomega$. Moreover, we will allow the set $\{0,
\dots, N-1\}$ to be partitioned into $d>0$ disjoint subsets $B_1, \dots, B_d$,
such that 
\begin{equation}
    \bigcup_{j=1}^d B_j = \{0, \dots, N-1\}.
\end{equation}
The central object of our work is the probability of observing a group $B_j$
when sampling from $\Dc_U$, namely,
\begin{equation}\label{eq:groupPr}
    \Pr[B_j] = \sum_{\substack{\s \in \Phi_m^n \\ \bomega \cdot \s \in B_j}} \Pr_{\s \sim \Dc_U}[\s].
\end{equation} 
This concept is illustrated in \cref{fig:linearStatistics}. Different
problems related to those probabilities can be defined depending on the task.
For example, one could be interested in estimating each of these probabilities
to inverse-polynomial additive precision or even sampling from the induced distribution.  We remark that, even though individual outcome probabilities of boson samplers tend to be exponentially small, groups $B_j$ that comprise a sufficiently large proportion of the outcome space can have sufficiently large probabilities (e.g., larger than $1/\text{poly}(m)$) and so inverse-polynomial additive precision is enough for a meaningful estimation of their value. 

\section{Classical simulation of linear statistics}\label{sec:results}

Having access to samples from an ideal Boson Sampling device, it is possible to estimate the probabilities $\text{Pr}[B_j]$ 
up to additive error $\epsilon$, using $O(1/\epsilon^2)$ samples. In this section, we show that a similar estimation, i.e. up to polynomial factors in the running time, can be done via efficient classical algorithms.  
Hereafter, we express linear statistics via a parameter-dependent linear
function $f_{\bomega} : \Phi_m^n \to \mathbb N$, defined as $f_{\bm
\omega}(\bm s) = \bomega \cdot \bm s$.
Let $N$ be a strict upper-bound on the maximum value $f_{\bomega}$ can take, i.e. $f_{\bomega}(\s)<N$ for every $\s\in\Phi_m^n$. A simple choice for this upper bound is
\begin{equation}
N=1+n\max_i|\omega_i|.
\end{equation}
Without loss of generality, we take
$N$ to be even, increasing it by one if necessary and retaining the notation $N$. We will see $f_{\bomega}$ as inducing an order in the event space $\Phi_m^n$ defined as 
\begin{equation}\label{sec:order-f}
    \s \preccurlyeq \t \Leftrightarrow f_{\bomega}(\s) \leq f_{\bomega}(\t),
\end{equation}
for any two outcomes $\s, \t \in \Phi_m^n$.  In  what follows, we assume that the groups are
contiguous in $\{0, \cdots, N-1\}$, namely, that the $j$-th group is defined as 
\begin{equation}\label{eq:groupBoundary}
    B_j = \{\bm s \ |\  x_{j-1} + 1 \leq f_{\bomega}(\bm s) \leq x_j\},
\end{equation}
so that the whole partition is described by the tuple $(x_0, x_1, \dots,
x_{d})$. 
Clearly, an efficient classical algorithm giving an estimation of $\Pr[B_j]$ should not directly sum up each event probability that contributes to the sum in \cref{eq:groupPr}, as there may be exponentially many such events.
Instead, the key idea of our algorithm to circumvent the direct computation of every probability directly
is to consider the Fourier decomposition of the
group probabilities, similarly to the approaches of \cite{seron_efficient_2024, oh_quantuminspired_2024}.

We first relate the characteristic function of the linear statistics distribution to a transition amplitude of a linear-optical circuit. Recall that, for a classical random variable $X$, the characteristic function is defined as

\begin{equation}
    \chi(t) = \e{e^{\imath tX}}.
\end{equation}

\begin{restatable}[Characteristic function of linear statistics
    \cite{arkhipov_2017}]{lm}{bsgen}
\label{thm:generatingFunction}
    Let $U\in \mathbb{U}(m)$, let $\bomega\in\mathbb{N}^m$, and let
    $\t \in\Phi_m^n$ be the input configuration. For every
    $\theta\in\mathbb{R}$, define
    \begin{equation}
    \chi_{\bomega}(\theta)
    =
    \mathbb{E}_{\s \sim D_U}
    \left[
    e^{\imath\theta f_{\bomega}(\s)}
    \right].
    \end{equation}
    Then
    \begin{equation}\label{eq:chiTransitionAmplitude}
    \chi_{\bomega}(\theta)
    =
    \langle\t|
    \hat V_{\Omega,\theta}
    |\t\rangle,
    \end{equation}
    where $\Omega=\operatorname{diag}(\omega_1,\ldots,\omega_m)$
    and $\hat V_{\Omega,\theta}$ is the linear-interferometer
    operator associated with the unitary matrix
    \begin{equation}
    V_{\Omega,\theta}
    =
    U^\dagger e^{\imath\theta\Omega}U.
    \end{equation}
\end{restatable}

 A proof is given in Appendix \ref{sec:CFApprox}. Since we assume $f_{\bomega}(\bm s)$ is integer-valued and bounded by $N$, in the discrete Fourier estimators used below, we only need to evaluate the 
characteristic function at the discrete set of values
\begin{equation}
\theta_k=-\frac{2\pi k}{N},
\end{equation}
 for $k\in\{0,\ldots,N-1\}$, and use the shorthand
$\chi_{\bomega}(k)=\chi_{\bomega}(\theta_k)$. Combined with the transition-amplitude approximation algorithm of \cite{lim_classical_2025}, \cref{thm:generatingFunction} implies that $\chi_{\bomega}(k)$ can be estimated efficiently with a classical computer to inverse-polynomial additive error in polynomial time for a large class of product input states. The characteristic function also allows us to express an estimator for the cumulative distribution. For $x \in \{0,...,N-1\}$, define the cumulative distribution function of a linear statistic as
\begin{equation}\label{eq:SxZero}
    S(x) = \Pr_{\bm s \sim \Dc_U}\left[f_{\bomega}(\s) \leq x\right].
\end{equation}
Clearly, an estimator for the cumulative distribution can be used to estimate the probabilities $\Pr[B_j]$. The following Lemma provides such an estimator based on the characteristic function.

\begin{restatable}[Fourier estimator of the cumulative distribution]{lm}{mcfourier}
\label{thm:estimateS}
Let $\varepsilon>0$, $0<\delta <1$. For $x \in\{0, \ldots, N-1\}$, the cumulative probability $S(x) $ is given by 
    \begin{equation}\label{eq:s-exp-val}
        S(x) = \e[k\sim q]{\frac{\chi_{\bomega}(k)\, G_N(k;x)}{N q(k)}},
    \end{equation}
    where $q$ is any probability distribution with support on $\{0, \cdots,
    N-1\}$ and
    \begin{equation}
        G_N(k;x) = \begin{cases}
            x + 1, 
                & k = 0, \\[6pt]
            \displaystyle \frac{1 - \exp\!\left(\frac{2\pi i k}{N}(x+1)\right)} {1 - \exp\!\left(\frac{2\pi i k}{N}\right)},
                & k \neq 0.
        \end{cases}
    \end{equation}
    Thus, provided $q$ can be sampled from and 
    \begin{equation}\label{eq:boundedModulus}
        \left| \frac{\chi_{\bomega}(k)\, G_N(k;x)}{N q(k)} \right| = O(R),
    \end{equation}
    $S(x)$ can be estimated to additive error $\varepsilon$ with failure
    probability $\delta$ using $\poly{\varepsilon^{-1}, R,
    \log\delta^{-1}}$ independent samples from $q$.
\end{restatable}

A detailed proof is given in \cref{sec:EstimateS}.
\cref{thm:estimateS} shows that efficient estimation of $S(x)$ reduces to finding a distribution $q$ that we can sample from efficiently, and for which the range of the estimator remains polynomially bounded. We use the following importance sampling distribution.
\begin{restatable}[Importance sampling distribution]{lm}{importance}\label{lm:Q}
    Let $N\in \mathbb N$ be even, and define
     \begin{equation}\label{eq:tk}
    t(k) =\begin{cases}
        1, & \text{if } k = 0 \\
        \frac{1}{k}, & \text{if } 1 \leq k \leq N/2 - 1, \\
        \frac{1}{N-k}, & \text{if } N/2 \leq k \leq N - 1, \\
    \end{cases}
\end{equation}
    and let $\mathcal N = \sum_{k=0}^{N-1} t(k)$ be the normalization factor. Then $q(k) = t(k)/\mathcal N$
    is a valid probability mass function from which samples can be generated classically in 
    expected time $\poly{\log N}$.
    Moreover, the range of the estimator in \cref{eq:boundedModulus} satisfies
    \begin{equation}
        \left| \frac{\chi_{\bomega}(k)\, G_N(k;x)}{N q(k)} \right| = O(\log N),
    \end{equation}
    uniformly in $x$ and $k$.
\end{restatable}
The preceding lemmas imply our main result.
\begin{restatable}[Classical estimation of cumulative distribution of linear statistics]{thm}{owf}\label{cor:OWF} 
    Consider a Boson Sampling instance with $n$ photons and $m = \poly{n}$ modes, specified by a unitary matrix
    $U\in \mathbb{U}(m)$ and an input state configuration $\t \in \Phi_m^n$. For a given  weight
    vector $\bomega \in \mathbb N^m_{\geq 0}$, let $f_{\bomega}: \Phi_m^n \to \{0, \dots, N-1\}$ be the function defining a linear statistic. For every $x \in \{0, \ldots, N-1\}$  there exists a classical randomized algorithm that estimates the cumulative probability
    $S(x)=\Pr_{s\sim D_U}[f_{\bomega} (\s)\leq x]$ to additive error $\varepsilon$ with failure probability at most $\delta$, in expected time
        \begin{equation}\label{eq:TimeComplexity}
            \poly{m, \log N, \varepsilon^{-1}, \log(1/\delta)}, 
        \end{equation}
        with $0 < \varepsilon, \delta < 1$.
        In particular, for any family of instances satisfying $\log N = \poly{m}$, the algorithm runs in expected time polynomial in $m$, $\varepsilon^{-1}$, and $\log(1/\delta)$. 
\end{restatable}
A detailed proof, including the error arising from the additive approximation of $\chi_{\bomega}(k)$, is given in \cref{sec:thm1}. \cref{cor:OWF} provides a general way to efficiently simulate linear statistics of boson samplers. The probabilities of observing outcomes in each group $B_j$ in \eqref{eq:groupBoundary} can be directly expressed as the difference between two cumulative distribution values
\begin{equation}\label{eq:prcontbins}
    \Pr\big[B_j\big] = S(x_j)-S(x_{j-1}).
\end{equation}
Estimating each of $S(x_j)$ and $S(x_{j-1})$ to additive error $\varepsilon/2$ yields an $\varepsilon$-additive estimate of $\Pr\big[B_j\big]$. The importance sampling distribution plays a crucial role, as it ensures the classical simulation algorithm remains efficient even in scenarios where the range of the function $f_{\bomega} (\s)$ is exponentially large as considered, for example,  in Refs.~\cite{nikolopoulos_decision_2016}.  We discuss applications of these techniques in the following section.

\section{Applications}

\begin{figure*}[t]
    \includegraphics[width = \textwidth]{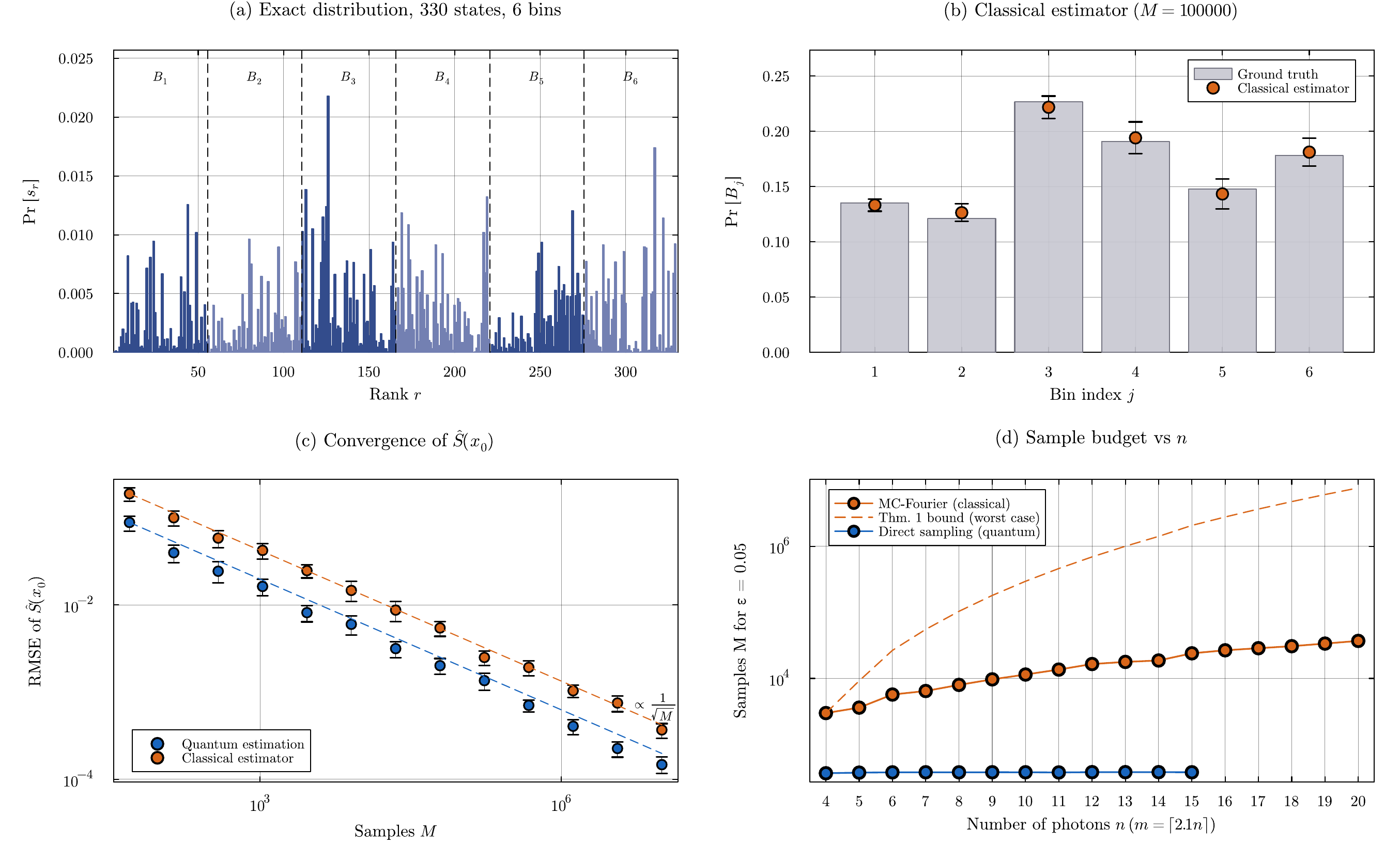}
    \caption{\textbf{Classical estimation of the boson-sampling output distribution via the
 algorithm presented in this work.} The input state $\ket*{1^n 0^{m-n}}$ evolves through an 
  interferometer defined by a Haar-random $U$, and outputs are encoded by $f_{\bomega}(s)=\sum_j s_j (n{+}1)^{j-1}$.
  \textbf{(a) Exact distribution and group partition.} Exact output probabilities for $n=4$, $m=8$
  ($330$ output configurations), ordered by $f_{\bomega}$ (see \cref{sec:order-f}) and partitioned
  into $d=6$ bins $B_1,\dots,B_6$ of equal size (represented with dashed lines).
  \textbf{(b) Classical group-probability estimates.} Group probabilities $\Pr[B_j]$ estimated via our algorithm
  (orange, $M=100{,}000$ samples per group, with errors obtained through batch means estimation with $K=40$ batches of $2\ 500$; error bars are
  $95\%$ confidence intervals) against the exact values
  (grey). The estimator correctly identifies the most-probable bin $B_3$.
  \textbf{(c) Convergence with sample size.} 
  Error from our estimation algorithm for $\hat S(x)$ to the exact value $S(x)=0.483$ found by direct summation versus the number of samples $M$, with $x$ at the boundary
  above group $3$. Both estimators decay as $O(M^{-1/2})$
  (dashed guides); the constant vertical offset is the $n^2\log^2 N$ variance
  penalty of our method relative to direct (quantum) sampling.
  \textbf{(d) Sample budget versus photon number.} Sample budget $M$ to reach additive error $\varepsilon=0.05$
  (95\% confidence interval) on a single bipartition of $\Phi_m^n$ in sets of equal size, versus photon number
  $n$ and $m=n^2$ modes. Direct quantum sampling (realised by the Clifford \& Clifford algorithm \cite{clifford_faster_2024}) is nearly constant in $n$
  whereas the sample needed for our method grows polynomially, significantly below the worst case bound $\propto n^2\log^2 N$ (dashed curve), with
  $\log N = n^2\log(n{+}1)$. 
    }
    \label{fig:owf}
\end{figure*}

In this section, we discuss the classical simulability of different applications of Boson Sampling relying on coarse-graining of the output distribution based on linear statistics. In \cref{sec:owf}, we apply our general result stated in \cref{cor:OWF} to provide the first efficient classical technique for simulating the most probable bin problem defined in \cite{nikolopoulos_decision_2016} and used in the one-way function proposal based on Boson Sampling from \cite{nikolopoulos_cryptographic_2019}.  Then in \cref{sec:mode-bins} we consider the related but distinct task of estimating a single probability of a linear statistic
and show how our result improves the existing classical method for estimating mode-grouped probabilities \cite{seron_efficient_2024}. We then argue that our
simulability results can be applied to boson sampling with different types of input states based on recent methods for estimating transition amplitudes in linear optics with product input states \cite{lim_classical_2025} (see \cref{sec:extensions} for details). This implies that our classical simulability results can be used to unify in the same framework of linear statistics results regarding simulability of binning strategies previously considered for
GBS, with applications in vibronic spectroscopy and validation of quantum devices
\cite{huh_yung_vibronic_2017,bressanini_gaussian_2024,oh_quantuminspired_2024}.
Additionally, in \cref{sec:extensions} we also explain how to take into account experimental imperfections such as partial photonic distinguishability and losses  at the level of the characteristic
function, without increasing the complexity of the classical simulation algorithm. We leave open the question of faster simulations in noisy scenarios.

\subsection{One-way function based on linear statistics}
\label{sec:owf}
A potential promising application of using a Boson Sampling device to solve a specific decision or function problem pertains to the design of
cryptographic one-way functions (OWFs)
\cite{nikolopoulos_decision_2016,nikolopoulos_cryptographic_2019} and so it is crucial to understand what kind of post-processing functions lead to probability distributions that can be simulated classically. An OWF $\Fc$
is a function that is easy to evaluate, but hard to invert on average. The core of the OWF proposed in
\cite{nikolopoulos_cryptographic_2019} is the estimation of the
largest group probability (often referred to as the \emph{most-probable bin} \cite{nikolopoulos_decision_2016}) according to \cref{eq:groupPr}, for a sequence of input states. We focus here on the classical simulability of this part of the protocol. The full protocol also involves further classical post-processing once the grouped probabilities are obtained, but it is not relevant to our discussion. 

Precisely, the quantum part of the protocol consists in linking an input state description $\bm t^* \in \Phi_{m}^n$ to the most probable bin $B^*$, amongst a partition $B_1, \cdots, B_d$ of the set of measurement outcomes $\Phi_m^n$. The partition defined in \cite{nikolopoulos_decision_2016} and used in the context of the OWF implementation \cite{nikolopoulos_cryptographic_2019, wang_experimental_2023} is based on the lexicographic ordering of the event space. The most natural way to induce this ordering is via the linear function 
\begin{equation}\label{eq:OWF}
    \Fc(\bm s) = f_{\bomega^{(\textsc{owf})}}(\bm s),
\end{equation}
with weights
\begin{equation}
    \bomega^{(\textsc{owf})} = (1, n+1, (n+1)^2, \dots, (n+1)^{m-1}),
\end{equation}
which is the $n$-ary number based on the digit representation of $\bm s = (s_1,\dots,s_m)$. In general, any weight vector with $\omega_{i+1}>n\omega_i$ orders the event space $\Phi_m^n$ in the same way, so the definition of $f_{\bomega^{(\textsc{owf})}}$ is not unique. To obtain the bins, outcomes $\s$ are separated in contiguous intervals with a similar number of outcomes. For example, $B_1$ contains the first $|\Phi_{m}^n|/d$ ranking outcomes and $B_2$ the next $|\Phi_{m}^n|/d$. When $d$ does not divide $|\Phi_m^n|$, the remaining are distributed among the first bins, as in \cite{nikolopoulos_decision_2016}, so that the bin size differs at most by one. 

An experiment may estimate the probability $\Pr[B_i]$ of each set of outcomes $B_i$, given an input $\bm t \in \Phi_m^n $ through direct sampling. 
The most probable bin is therefore estimated as 
\begin{equation}
    \widehat{\text{MPB}}(\bm t) = \arg\max_i \widehat{\Pr}[B_i]
\end{equation}
where $\widehat{\Pr}[B_i]$ is the estimate for $\Pr[B_i]$. This particular binning strategy allows us to use our method for evaluating the cumulative distribution of linear statistics to estimate each probability.  Hence, the most probable bin can be estimated with a number of samples only up to a $\poly{\log N}$ factor with respect to direct sampling.

Given that an experiment is limited to $\poly{m}$ samples, the number of bins $d$ must itself be kept to $d = \poly{m}$ to obtain meaningful estimates $\widehat{\Pr}[B_i]$. The groups $B_1, \dots, B_d$
are thus associated with a list of integers $(0, x_1, \dots, x_{d-1},
N^{\textsc{(owf)}})$ describing the \emph{boundaries} of the groups in
accordance with \cref{eq:groupBoundary}, where we set 
\begin{equation}
    N^{(\textsc{owf})} = n(n+1)^{m-1} + 2,
\end{equation}
so that $f_{\bomega^{(\textsc{owf})}} (\s) < N^{(\textsc{owf})}$. Since $n \leq m$, we have that
\begin{equation}
    \log N^{(\textsc{owf})} = O(m\log (n+1)) = \poly{m}.
\end{equation}
Therefore, by \cref{cor:OWF} we can estimate each bin probability classically in expected polynomial time, with the same level of precision as from polynomially many samples obtained through an actual experiment. We illustrate our technique with numerical simulations in \cref{fig:owf}. We stress that, if we take into account the effect of loss in the experimental hardware, the  time required to obtain
  samples from the quantum device actually grows exponentially with the number of photons, since the probability that all photons arrive to the output decreases exponentially.

While our method requires asymptotically more samples than an ideal experiment does, we stress that the running time for classical estimation for finite-size experiments can be significantly shorter. For instance, we are able to estimate the cumulative distribution $S(x)$ for $m = n = 100$ within a $0.01$ additive-error in about five seconds on average on a standard laptop \cite{seron2024bosonsampling}. Indeed, this number of photons is largely above the experimental capabilities of current hardware due to loss. 

It is important to remark that, in the context of cryptographic applications,  it is relatively simple to come up with strategies that circumvent our classical simulability results by considering other kinds of coarse-graining functions that do not directly map to linear statistics. One proposal for a hash function from Ref.~\cite{shi_quantum_2022} treats differently collision-free events from events with collision, which amounts to a non-linear postprocessing of the outcomes. Moreover, in the quantum proof-of-work protocol proposed in \cite{singh_proofofwork_2025}, one step of the protocol involves sorting the Boson Sampling outcomes according to its $n$-ary decomposition, applying a permutation to shuffle the ordering, and then binning the probability distribution. After the permutation, the new ordering obtained for the outcomes is not, in general, given by a linear function of the entries of $\s$. We leave open the question of whether these binning strategies are prone to efficient classical simulation methods.  

\subsection{Improved-time complexity for mode-binning}
\label{sec:mode-bins}

Grouped probabilities are also at the core of techniques of validation of Boson
Sampling \cite{seron_efficient_2024} and Gaussian Boson Sampling
\cite{bressanini_gaussian_2024}. In these techniques, the $m$ modes are
partitioned into $K$ groups $\Kc_1, \ldots, \Kc_K$, and one is interested in the
number of photons in each group of modes, rather than in each mode. For the case of Boson Sampling, an
algorithm is given in \cite{seron_efficient_2024} for estimating the probability
of a group occupation $(k_1, \ldots, k_K)$, i.e., with $k_i$ photons in the
group $\Kc_i$, in time $n^{O(K)}$. That is, the algorithm is efficient when
$K=O(1)$. After describing how this particular grouping strategy can be
expressed via linear statistics, we show how the group probabilities can be
estimated efficiently, even when the number of groups is as large as the number of output modes. We consider the weight vector
$\bomega^{\textsc{(mg)}}$ corresponding to mode-grouping 
\begin{equation}
    \label{eq:omega_spatial_binning}
    \hspace{-.8em}\bomega^{\textsc{(mg)}}\! =\! (\underbrace{1, \cdots, 1}_{|\Kc_1|\ \text{times}}, 
    \dots
    , \underbrace{(n+1)^{\textsc{k}-1}, \cdots, (n+1)^{\textsc{k}-1}}_{|\Kc_K|\ \text{times}}
    ),
\end{equation}
where the entry $(n+1)^i$ is repeated $|\Kc_{i+1}|$ times. Thus, observe that
\begin{align}
    f_{\bomega^{\textsc{(mg)}}}(\s) 
         = \bomega^{\textsc{(mg)}} \cdot \s 
         = \sum_{i = 1}^K\ (n+1)^{i-1} k_i,
\end{align}
where $k_i = \sum_{j \in \Kc_i} s_j$, corresponds to the total number of photons
in the $i$-th group. Indeed, this choice of weight-vector collects the number of
photons in each group of modes, so that all measurement outcomes corresponding
to the same group occupation are associated with the same integer. For a
mode-group occupation $\bm\ell \in \Phi_K^n$, indicating that $\ell_i$ photons
occupy the $i$-th group, we denote by $\mathscr{O}_{\bm\ell}\subseteq \Phi_m^n$
the set of photon occupations whose mode-grouping is $\bm \ell$, that is,
\begin{equation}
   \! \mathscr{O}_{\bm\ell} = \left\{\s \ | \ \s \in \Phi_m^n,\, \textstyle \sum_{j \in \Kc_i} s_j = \ell_i, \forall\ 1 \leq i \leq K\right\}.
\end{equation}

We obtain the following \cref{cor:modeBinning}.

\begin{restatable}{cor}{modebins}
\label{cor:modeBinning}
    Let $\kappa =
    \sum_{i=1}^{\textsc{k}} (n+1)^{i-1} k_i$ with ${k_1, \cdots, k_K \in \Nb}$.
    The probability of observing the group occupation ${\bm k = (k_1, \cdots,
    k_\textsc{k})}$, i.e., $k_i$ photons in $\Kc_i$ can be estimated classically
    to additive error $\epsilon>0$ via the linear statistics 
    \begin{equation}
        \Pr_{\s \sim \Dc_U}[f_{\bomega^{(\textsc{mg})}}(\s) = \kappa] = \sum_{\bm s \in \mathscr{O}_{\bm k}} \Pr[\s],
    \end{equation}
    with weight vector $\bomega^{(\textsc{mg})}$ and $q$ the uniform
    distribution, in time $O(m^2/\epsilon^4)$ with high probability.
\end{restatable}

See \cref{sec:proofModes} for more details. When $K = m$, so that group
probabilities correspond to actual outcome probabilities, our algorithm recovers
the time complexity of the algorithm of \cite{lim_classical_2025}. Importantly, while previous work focused on obtaining the full distribution, thus leading to an exponential dependency in the number of groups \cite{seron_efficient_2024}, \cref{cor:modeBinning} gives a classical algorithm
that lifts this dependency as we are only interested in estimating a  single probability of a given mode-group occupation. Moreover, special choices of
$\bomega^{(\textsc{mg})}$ thus allow one to estimate any marginal of $\Dc_U$.

We remark that while this section focuses on recovering mode-binned distributions of \cite{seron_efficient_2024}, other choices of the weight vector $\bomega$ which do not necessarily map directly to spatial mode binning could also be used in the context of validation of boson samplers. It remains to be explored whether more general tests based on linear statistics could provide more stringent validation techniques for Boson Sampling beyond simple mode-binning. 

\bs{I think that this statement, though true, is slightly misleading. The exponential
dependency in K was due to the fact that we compute the full probability distribution 
instead of accessing a single element (like you do here). It's not obvious whether we could
have cherry picked a single element: there, you clearly bring something new.
I don't really care about being precise here and you could leave it like that if it helps
the paper's narrative}\hft{I think we could remove the dependency on K in you paper by using Monte Carlo on eq.26 on your paper, published version. It's just that it is not mentionned. }
\ea{In the abstract I wrote we have an improved time algorithm to compute single group-mode probabilities, so this discussion will have repercussions there. Personally, I think we should specify that the removal of the dependency is in the estimation of single group-mode probabilities if indeed in Benoits paper they depend exponentially on K}\hft{I'm not sure we should say that in the abstract as it's very technical. We should discuss about that together.}

\subsection{Extension to other sampling schemes and partially distinguishable photons}
\label{sec:extensions}

In this section, we highlight the role of the characteristic function as the
only physics-dependent quantity. That is to say, the dependence on the input
state and the transformation is only visible in the characteristic function.
Therefore, our \cref{cor:OWF} can be adapted to any experimental setup
giving rise to a characteristic function that can be either computed exactly or
approximated within small error efficiently classically. We detail two
different extensions.

\paragraph{Gaussian Boson Sampling.}

Given a unitary matrix $U \in \Ub(m)$, a weight vector $\bomega \in \Nb^m$ and
an input displaced squeezed vacuum of the form $\ket{\psi_{in}} = \hat
D(\bm{\alpha})\hat S(\bm r)\ket{\bm{0}}$. Using positive P-representation \cite{oh_quantuminspired_2024,bressanini_gaussian_2024}, the characteristic function of linear statistics of the distribution $\Dc_U$
reads

\begin{equation}
    \chi^{\textsc{(gbs)}}_{\bomega}(k)
        = \mathcal N \sqrt{\frac{(2\pi)^{2m}}{|Q|}}\exp(\frac12 \bm \zeta^T Q^{-1}\bm \zeta + \bm\zeta_0),
\end{equation}
where 
\begin{equation}
    Q = \begin{bmatrix}
        2 \Gamma^{-1}+ \mathbb I_m & -U^T \exp(\imath \pi k \Omega) U^* \\
        -U^\dag \exp(\imath \pi k \Omega) U & 2 \Gamma^{-1}+ \mathbb I_m
    \end{bmatrix},
\end{equation}
and 
\begin{equation}
\begin{aligned}
    &\gamma_i = \exp(2r_i) - 1, \quad 
    \mathcal N = \prod_{i=1}^m\frac{\sqrt{1+\gamma_i}}{\pi \gamma_i},\\
    &\Gamma = \text{diag}(\{\gamma_i\}_{i=1}^m), \quad
    \Phi = \exp(\imath \pi k \Omega) - \mathbb{I}, \\
    &\bm \zeta = \frac{1}{\sqrt2}\begin{bmatrix}
        U^T\Phi \bm \delta^* \\ U^\dag \Phi \bm \delta
    \end{bmatrix}, \quad
      \bm\zeta_0 = \frac12 \bm \delta^T \Phi \bm \delta^*.
\end{aligned}
\end{equation}
That is, the characteristic function of displaced
squeezed vacuum can be computed exactly in polynomial time, and a similar result
holds for a thermal input state \cite{bressanini_gaussian_2024}. The fact that the characteristic function for displaced squeezed input states can be computed exactly classically is at the heart of the quantum-inspired classical algorithm for approximating molecular vibronic spectra \cite{oh_quantuminspired_2024}, which is given by linear statistics of Boson Sampling with this type of inputs \cite{huh_vibronic_2015}. 

Moreover, by using this characteristic function it is also possible to use the results of Sec.~\ref{sec:mode-bins} for computations of mode-binned distributions for Gaussian Boson Sampling, from which marginal distributions and photon-number correlations can be obtained, which are important for device validation
\cite{fitzke_simulating_2023,cardin_photonnumber_2024, drummond2022simulating,bressanini_gaussian_2024}.

On the other hand, the characteristic function of linear statistics of GBS with squeezed Fock input states
is now related to a matrix Hafnian
\cite{hamilton_gaussian_2017}, for which an efficient classical estimation
algorithm is known \cite{lim_classical_2025}.
Therefore, \cref{cor:OWF} also extends  naturally to these settings. 
An important consideration that needs to be taken into account is a photon number truncation, which is possible for the types of states discussed before (see \cite{bressanini_gaussian_2024,upreti_exponentiallyimproved_2026} for discussions). This truncation will affect the upper bound on the range of the function $f(\s)$ and thus on the number of Fourier coefficients appearing in the estimator of $S(x) $ in \cref{eq:s-exp-val}. However, since the error induced by the truncation typically decreases exponentially with $N$, this does not significantly affect the performance of the algorithm.

\paragraph{Experimental errors.}
The two main sources of errors in photonic experiments are losses and partial distinguishability of the input photons. In the case of Fock Boson Sampling, outcome probabilities are generally post-selected on finding all photons at the output, but events are still affected by partial distinguishability, which may render the outcome probabilities
easier to compute. For pure states, this effect can be described by the so-called Gram matrix ${(S)_{ij} =
\braket{\phi_i}{\phi_j}}$ of overlaps of internal states of photons \cite{seron_efficient_2024}. Recall from
\cref{thm:generatingFunction} that the characteristic function can be written as
the normalized matrix permanent of $V_{\bm s}= (U^\dagger \exp(-\tfrac{2\pi \imath k}{N}
\Omega) U)_{\s,\s}$ with $\Omega=\text{diag}(\omega_1, \cdots, \omega_m)$ and
$\s$ the input occupation. It was shown \cite{seron_efficient_2024} that the
characteristic function associated with partially distinguishable photons reads
\begin{equation}\label{eq:charPartialDist}
    \chi_{\bomega}^{(\textsc{pd})}(k) \propto \per{S \odot V_{\bm s}},
\end{equation}
where $\odot$ is the Hadamard (element-wise) product defined as $(A \odot
B)_{ij} = (A)_{ij}(B)_{ij}$. Therefore, partial distinguishability can be
directly encoded into the characteristic function, which, plugged into
\cref{cor:OWF} yields an efficient algorithm for estimating linear
statistics in the presence of partial distinguishability.
In general, losses can easily be integrated by encompassing the $m-$mode network $U$ in a larger one with $2m$ modes \cite{seron_efficient_2024}, where the additional bottom $m$ modes correspond to lost photons and are traced out. This fits naturally in our formalism through an appropriate selection of $f$.

\section{Simulating photon-photon interactions via non-linear statistics}
As previously mentioned, certain binning strategies proposed in \cite{shi_quantum_2022} and in \cite{singh_proofofwork_2025} cannot directly be encompassed in our linear statistics framework and thus it is an open question whether they can be classically simulated efficiently. However, these strategies are motivated solely by cryptographic applications. As a first step towards understanding what types of coarse-graining of the Boson Sampling can be related to difficult physics-motivated problems, let us consider a simple proposal for \emph{nonlinear statistics} described via a quadratic
form as $f_{n.l.}(\s) = \sum_{j, k } J_{jk} s_j s_k $. For simplicity, let us consider the coefficients are integers and upper bounded by some positive value $C$. Clearly, in this scenario $f(\s) \leq C n^2$ and so a polynomial number of samples from the Boson Sampling device would be sufficient to estimate all the probabilities $P_k = \Pr_{\s \sim \Dc_U}[f_{n.l.}(\s) = k]$. By taking a discrete Fourier transform of these values, it is possible to obtain points of the characteristic function, which can be seen as transition amplitudes of a process with one layer of photon-photon interaction (Kerr and cross-Kerr interactions). Precisely, for this type of non-linear statistic, the analogue of \cref{eq:chiTransitionAmplitude} takes the following form 
    \begin{align}
        \chi_{\bomega}(k)
            & = \e[\s \sim \Dc_U]{\exp\!\lp\!-\frac{2\pi \imath k}{N} f_{n.l.}(\bm s)\rp } \label{eq:chiExpVal_non_linear}
            = \mel{\t}{\hat V_{\Omega_{int}}}{\t},
    \end{align}
    where 
    \begin{align}
       \hat{V}_{int} &= \hat{U}^{\dagger} \exp(-\frac{2\pi \imath k}{N} \hat{H}_{int})\hat{U}, \\
       \hat{H}_{int}&=  \sum_{j, k } J_{jk} \hat{n}_j \hat{n}_k.
    \end{align}
 In general, such transition amplitudes are $\#\text{P}$-hard to estimate to exponentially small additive error \cite{jabbour2025complexity} since they can be used to estimate individual outcome probabilities of Boson Sampling with exponential precision. However, to our knowledge there is currently no efficient classical algorithm to estimate them to inverse polynomial precision, which is the precision we would be able to get from polynomially many samples of the Boson Sampling device. State-of-the-art techniques to approximate amplitudes of processes involving Kerr and cross-Kerr interactions, based on expanding these non-linear gates in terms of linear combinations of phase-shifters, only allow for efficient additive error estimates when the number of such terms scales logarithmically in the number of modes \cite{upreti_exponentiallyimproved_2026}. In contrast, there are $O(m^2)$ non-linear gates in \cref{eq:chiExpVal_non_linear} and we leave open whether an efficient (Gurvits-like) additive error estimation algorithm is possible for these quantities. 
 
 More generally, as depicted in \cref{fig:physical-process}, the probability distribution obtained by coarse-graining the Boson Sampling outcome distribution via a non-linear function can be related (via discrete Fourier transform) to transition amplitudes of linear-interference processes which also contain a non-linear Hamiltonian evolution associated with an interacting Hamiltonian $H_{int}=f(\hat{n}_1, ..., \hat{n}_m)$. Hence, the question of classical simulability of non-linear
statistics is tightly related to transition-amplitude estimation
in circuits with a single layer of non-linear gates, as considered
in non-linear Boson Sampling (NLBS) \cite{spagnolo2023non}.
 
 For comparison, allowing polynomially many layers of cross-Kerr
interactions interleaved with linear optics enables universal
quantum computation \cite{lloyd1999quantum}.
Consequently, even estimating a single-mode marginal probability
to inverse-polynomial additive error is BQP-hard in this regime,
so an efficient classical algorithm would imply
$\mathrm{BPP}=\mathrm{BQP}$. It remains an open question how the complexity of additive 
estimation depends on the number of nonlinear layers, and whether
such hardness of classical simulations already arises for the single-layer scenario considered here.

    \begin{figure}
    \centering
    \includegraphics[width=\linewidth]{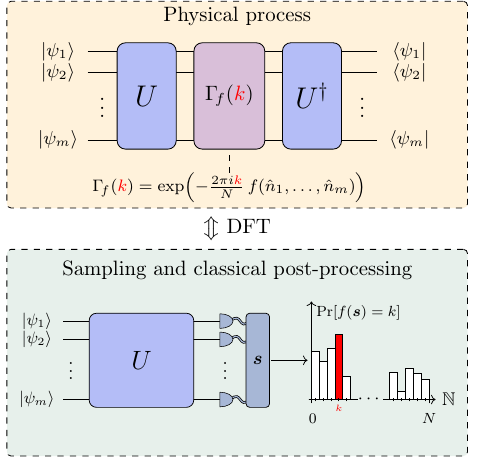}
    \caption{Coarse-grained output distributions of boson sampling $\Pr[f(\s) = k]$ after applying an arbitrary function $f(\s)$ to map the outcomes to integers (bottom). This probability distribution can be related to the transition amplitude associated to the physical process depicted on top -- one can be seen as the characteristic function of the other and so they are related by a discrete Fourier transform (DFT). \cref{thm:generatingFunction} formalizes this result for a linear function, but this results holds for arbitrary functions. For quadratic functions, $\Gamma_f(k)$ is given by a layer or Kerr gates. }
    \label{fig:physical-process}
\end{figure}

\section{Discussion}

In this work, we give a classical algorithm for efficiently estimating
\emph{linear statistics} of Boson Sampling, that is, the probability of obtaining a certain integer weighted
combination of the outcome's photon occupation. The methods work for different kinds of input states and allow for the classical simulation of molecular vibronic spectra, binned-mode distributions and certain proposals for one-way functions. Clearly, the method requires knowledge of the weight vector $\bm{\omega}$ to predict the linear statistics distribution. Certain cryptographic applications or validation tests using linear statistics require that the samples are sent before the weight vector is chosen and the coarse-grained distribution is analysed. It is unknown whether in these scenarios it would be possible to design an efficient mock-up classical sampler that would correctly mimic any linear statistics, for any random choice of weight vector. Since linear statistics encompasses mode-binning as a special case, it provides more stringent validation tests for Boson Sampling.

Our framework can also be extended to estimate probabilities of joint linear statistics, such as $\Pr[\bomega^{(1)}\cdot \s = k \land \bomega^{(2)}\cdot \s = j ]$ for different weight vectors $\bomega^{(1)}, \bomega^{(2)}$, or joint cumulative distributions involving
several integer-valued linear statistics.
For a fixed total photon number, the values of each linear function be shifted if
necessary and encoded in a sufficiently large integer base, reducing
simultaneous linear equalities to a single linear-statistic equality. Alternatively, a multidimensional characteristic function can be considered, similarly to the approach of \cite{seron_efficient_2024}. This preserves efficient additive-error estimation whenever the
logarithms of the individual range bounds are polynomial in the
number of modes.

Our work raises the question of finding other types of statistics that can be efficiently estimated with an
actual quantum device, but for which no known efficient classical algorithm
exists. Indeed, our methods apply directly when the outcomes are ordered via a linear function (e.g. by assigning an $n$-ary number) and the bins are chosen contiguously according to this ordering. More complex strategies to assign states to bins, such as shuffling the order via a permutation \cite{singh_proofofwork_2025} or treating collision-free outcomes differently from events with collisions \cite{shi_quantum_2022}, goes beyond our simulability regime. 

It would be interesting if non-linear statistics of Boson Sampling could find applications beyond cryptography and be used to solve more structured problems. In this direction, we show that the distributions obtained by coarse-graining the Boson Sampling data via a quadratic function can be used to estimate amplitudes of specific photonic quantum circuits which contain linear interferometers and one layer of photon-photon interactions, which escape current classical simulation techniques \cite{upreti_exponentiallyimproved_2026}.  We hope this connection instigates further investigations on the complexity of coarse-grained Boson Sampling distributions and their potential role in practical quantum advantage proposals.

\begin{acknowledgments}
The authors thank Andreas Buchleitner, Ulysse Chabaud, Saleh Rahimi-Keshari, Gabriel Dufour, Christoph Dittel, Tim Ehret, Pierre-Emmanuel Emeriau, Jelmer Renema, Stefan van den Hoven, Gyan Lamers, Innes Maxwell and Rawad
Mezher for fruitful discussions. H.T. acknowledges support from the European Commission as part of the EIC accelerator
program under the grant agreement 190188855 for SEPOQC project, the Horizon-CL4
program under the grant agreement 101135288 for EPIQUE project, and the CIFRE grant N$^o$ 2023/1746.
BS and LN acknowledge support from Horizon Europe project EPIQUE (Grant No. 101135288). B.S. acknowledges support from the Georg H Endress Foundation. E.A. acknowledges funding from FCT-
Fundação para a Ciência e a Tecnologia (Portugal). L.N. acknowledges funding from FCT-
Fundação para a Ciência e a Tecnologia (Portugal) via
the Project No. CEECINST/00062/2018 and from the European Union’s Horizon Europe Framework Programme (EIC Pathfinder Challenge
project Veriqub) under Grant Agreement No. 101114899.  
C.O. was supported by the National Research Foundation of Korea Grants (No. RS-2024-00431768 and No. RS-2025-00515456) funded by the Korean government (Ministry of Science and ICT (MSIT)) and the Institute of Information \& Communications Technology Planning \& Evaluation (IITP) Grants funded by the Korean government (MSIT) (No. RS-2024-00437284, No. IITP-2025-RS-2025-02283189 and No. IITP-2025-RS-2025-02263264). This work was supported by Global Partnership Program of Leading Universities in Quantum Science and Technology (RS-2025-08542968) through the National Research Foundation of Korea~(NRF) funded by the Korean government (Ministry of Science and ICT (MSIT)).
\end{acknowledgments}

\section*{Code availability}

The code is available at the following GitHub repository \bs{set this up}\cite{seron2024bosonsampling}.

\bibliography{bibliography}

@article{bradler2018gBS_perfect_matchings,
  title={Gaussian boson sampling for perfect matchings of arbitrary graphs},
  author={Br{\'a}dler, Kamil and Dallaire-Demers, Pierre-Luc and Rebentrost, Patrick and Su, Daiqin and Weedbrook, Christian},
  journal={Physical Review A},
  volume={98},
  number={3},
  pages={032310},
  year={2018},
  doi = {10.1103/PhysRevA.98.032310},
  publisher={APS}
}

@article{seron2024bosonsampling,
  title   = {{BosonSampling.jl}: A {Julia} package for quantum multi-photon interferometry},
  author  = {Seron, Benoit and Restivo, Antoine},
  journal = {Quantum},
  volume  = {8},
  pages   = {1378},
  year    = {2024},
  doi     = {10.22331/q-2024-06-18-1378}
}

@misc{monbroussou_classical_2026,
  title = {Classical Simulation and Model Concentration in Passive Linear Optics},
  author = {Monbroussou, L{\'e}o and Thomas, Hugo and Mhiri, Hela and Holmes, Zo{\"e} and Kashefi, Elham},
  year = 2026,
  month = jul,
  number = {arXiv:2607.24728},
  eprint = {2607.24728},
  publisher = {arXiv},
  doi = {10.48550/ARXIV.2607.24728},
  urldate = {2026-08-07},
  archiveprefix = {arXiv},
  copyright = {Creative Commons Attribution 4.0 International}
}

@article{jabbour2025complexity,
  title={Complexity of Gaussian quantum optics with a limited number of non-linearities},
  author={Jabbour, Michael G and Novo, Leonardo},
  journal={Quantum Science and Technology},
  volume={10},
  number={4},
  pages={045021},
  year={2025},
  doi = {10.1088/2058-9565/adf6d4},
  publisher={IOP Publishing}
}

@article{anguita_experimental_2025,
  title   = {Experimental validation of boson sampling using detector binning},
  author  = {Correa Anguita, Malaquias and Camillini, Anita and Marzban, Sara
             and Robbio, Marco and Seron, Benoit and Novo, Leonardo
             and Renema, Jelmer J.},
  journal = {Quantum Science and Technology},
  volume  = {10},
  number  = {3},
  pages   = {035062},
  year    = {2025},
  doi     = {10.1088/2058-9565/adecbc}
}

@article{drummond2022simulating,
  title={Simulating complex networks in phase space: Gaussian boson sampling},
  author={Drummond, Peter D and Opanchuk, Bogdan and Dellios, Alexander and Reid, Margaret D},
  journal={Physical Review A},
  volume={105},
  number={1},
  pages={012427},
  year={2022},
  publisher={APS},
  doi = {10.1103/PhysRevA.105.012427}
}

@article{hangleiter2023computational,
  title={Computational advantage of quantum random sampling},
  author={Hangleiter, Dominik and Eisert, Jens},
  journal={Reviews of Modern Physics},
  volume={95},
  number={3},
  pages={035001},
  year={2023},
  doi = {10.1103/RevModPhys.95.035001},
  publisher={APS}
}

@article{clifford_faster_2024,
  title = {Faster Classical Boson Sampling},
  author = {Clifford, Peter and Clifford, Rapha{\"e}l},
  year = 2024,
  month = jun,
  journal = {Physica Scripta},
  volume = {99},
  number = {6},
  pages = {065121},
  issn = {0031-8949, 1402-4896},
  doi = {10.1088/1402-4896/ad4688},
  urldate = {2026-06-04}
}

@article{salavrakos_photonnative_2025,
  title = {Photon-Native Quantum Algorithms},
  author = {Salavrakos, Alexia and Maring, Nicolas and Emeriau, Pierre-Emmanuel and Mansfield, Shane},
  year = 2025,
  month = apr,
  journal = {Materials for Quantum Technology},
  volume = {5},
  number = {2},
  pages = {023001},
  publisher = {IOP Publishing},
  issn = {2633-4356},
  doi = {10.1088/2633-4356/adc531},
  urldate = {2025-05-26},
  langid = {english}
}

@article{maring_versatile_2024,
  title = {A Versatile Single-Photon-Based Quantum Computing Platform},
  author = {Maring, Nicolas and Fyrillas, Andreas and Pont, Mathias and Ivanov, Edouard and Stepanov, Petr and Margaria, Nico and Hease, William and Pishchagin, Anton and Lema{\^i}tre, Aristide and Sagnes, Isabelle and Au, Thi Huong and Boissier, S{\'e}bastien and Bertasi, Eric and Baert, Aur{\'e}lien and Valdivia, Mario and Billard, Marie and Acar, Ozan and Brieussel, Alexandre and Mezher, Rawad and Wein, Stephen C. and Salavrakos, Alexia and Sinnott, Patrick and Fioretto, Dario A. and Emeriau, Pierre-Emmanuel and Belabas, Nadia and Mansfield, Shane and Senellart, Pascale and Senellart, Jean and Somaschi, Niccolo},
  year = 2024,
  month = jun,
  journal = {Nature Photonics},
  volume = {18},
  number = {6},
  pages = {603--609},
  publisher = {{Springer Science and Business Media LLC}},
  issn = {1749-4885, 1749-4893},
  doi = {10.1038/s41566-024-01403-4},
  urldate = {2025-08-05},
  copyright = {https://creativecommons.org/licenses/by/4.0},
  langid = {english}
}

@phdthesis{arkhipov_2017,
  title  = {Quantum computation with identical bosons},
  author = {Arkhipov, Alex (Aleksandr)},
  school = {Massachusetts Institute of Technology},
  year   = {2017},
  url    = {https://hdl.handle.net/1721.1/113995}
}

@article{nikolopoulos_cryptographic_2019,
  title = {Cryptographic One-Way Function Based on Boson Sampling},
  author = {Nikolopoulos, Georgios M.},
  year = 2019,
  month = jul,
  journal = {Quantum Information Processing},
  volume = {18},
  number = {8},
  pages = {259},
  issn = {1573-1332},
  doi = {10.1007/s11128-019-2372-9},
  urldate = {2026-02-24},
  langid = {english}
}

@article{nikolopoulos_decision_2016,
  title = {Decision and Function Problems Based on Boson Sampling},
  author = {Nikolopoulos, Georgios M. and Brougham, Thomas},
  year = 2016,
  month = jul,
  journal = {Physical Review A},
  volume = {94},
  number = {1},
  pages = {012315},
  publisher = {American Physical Society},
  doi = {10.1103/PhysRevA.94.012315},
  urldate = {2026-02-24}
}

@article{singh_proofofwork_2025,
  title = {Proof-of-Work Consensus by Quantum Sampling},
  author = {Singh, Deepesh and Muraleedharan, Gopikrishnan and Fu, Boxiang and Cheng, Chen-Mou and Roussy Newton, Nicolas and Rohde, Peter P and Brennen, Gavin K},
  year = 2025,
  month = feb,
  journal = {Quantum Science and Technology},
  volume = {10},
  number = {2},
  pages = {025020},
  publisher = {IOP Publishing},
  issn = {2058-9565},
  doi = {10.1088/2058-9565/adae2b},
  urldate = {2026-02-24},
  langid = {english}
}

@article{wang_experimental_2023,
  title = {Experimental {{Boson Sampling Enabling Cryptographic One-Way Function}}},
  author = {Wang, Xiao-Wei and Zhou, Wen-Hao and Fu, Yu-Xuan and Gao, Jun and Lu, Yong-Heng and Chang, Yi-Jun and Qiao, Lu-Feng and Ren, Ruo-Jing and Jiang, Ze-Kun and Jiao, Zhi-Qiang and Nikolopoulos, Georgios M. and Jin, Xian-Min},
  year = 2023,
  month = feb,
  journal = {Physical Review Letters},
  volume = {130},
  number = {6},
  pages = {060802},
  publisher = {American Physical Society},
  doi = {10.1103/PhysRevLett.130.060802},
  urldate = {2026-02-24}
}

@article{aaronson_generalizing_2014,
  title = {Generalizing and Derandomizing {{Gurvits}}'s Approximation Algorithm for the Permanent},
  author = {Aaronson, Scott and Hance, Travis},
  year = 2014,
  month = may,
  journal = {Quantum Info. Comput.},
  volume = {14},
  number = {7\&8},
  pages = {541--559},
  issn = {1533-7146},
  doi = {10.26421/QIC14.7-8-1}
}

@article{oh_quantuminspired_2024,
  title = {Quantum-Inspired Classical Algorithms for Molecular Vibronic Spectra},
  author = {Oh, Changhun and Lim, Youngrong and Wong, Yat and Fefferman, Bill and Jiang, Liang},
  year = 2024,
  month = feb,
  journal = {Nature Physics},
  volume = {20},
  number = {2},
  pages = {225--231},
  publisher = {Nature Publishing Group},
  issn = {1745-2481},
  doi = {10.1038/s41567-023-02308-9},
  urldate = {2026-02-18},
  copyright = {2024 The Author(s), under exclusive licence to Springer Nature Limited},
  langid = {english}
}

@article{seron_efficient_2024,
  title = {Efficient Validation of {{Boson Sampling}} from Binned Photon-Number Distributions},
  author = {Seron, Benoit and Novo, Leonardo and Arkhipov, Alex and Cerf, Nicolas J.},
  year = 2024,
  month = sep,
  journal = {Quantum},
  volume = {8},
  pages = {1479},
  publisher = {Verein zur Forderung des Open Access Publizierens in den Quantenwissenschaften},
  issn = {2521-327X},
  doi = {10.22331/q-2024-09-19-1479},
  urldate = {2025-08-04},
  copyright = {https://creativecommons.org/licenses/by/4.0/},
  langid = {english}
}

@inproceedings{aaronson_computational_2011,
  title = {The Computational Complexity of Linear Optics},
  booktitle = {Proceedings of the Forty-Third Annual {{ACM}} Symposium on {{Theory}} of Computing},
  author = {Aaronson, Scott and Arkhipov, Alex},
  year = 2011,
  month = jun,
  pages = {333--342},
  publisher = {ACM},
  address = {San Jose California USA},
  doi = {10.1145/1993636.1993682},
  urldate = {2024-07-18},
  isbn = {978-1-4503-0691-1},
  langid = {english}
}

@article{hoeffding_class_1948,
  title   = {Probability Inequalities for Sums of Bounded Random Variables},
  author  = {Hoeffding, Wassily},
  journal = {Journal of the American Statistical Association},
  volume  = {58},
  number  = {301},
  pages   = {13--30},
  year    = {1963},
  doi     = {10.1080/01621459.1963.10500830}
}

@article{bressanini_gaussian_2024,
  title   = {Binned-detector probability distributions for {Gaussian} boson sampling validation},
  author  = {Bressanini, Gabriele and Seron, Benoit and Novo, Leonardo
             and Cerf, Nicolas J. and Kim, M. S.},
  journal = {Physical Review A},
  volume  = {112},
  number  = {1},
  pages   = {012610},
  year    = {2025},
  doi     = {10.1103/jqvf-pm1r}
}

@article{hamilton_gaussian_2017,
  title   = {Gaussian Boson Sampling},
  author  = {Hamilton, Craig S. and Kruse, Regina and Sansoni, Linda and Barkhofen, Sonja and Silberhorn, Christine and Jex, Igor},
  journal = {Physical Review Letters},
  volume  = {119},
  number  = {17},
  pages   = {170501},
  year    = {2017},
  doi     = {10.1103/PhysRevLett.119.170501}
}

@article{huh_vibronic_2015,
  title   = {Boson Sampling for Molecular Vibronic Spectra},
  author  = {Huh, Joonsuk and Guerreschi, Gian Giacomo and Peropadre, Borja and McClean, Jarrod R. and Aspuru-Guzik, Al{\'a}n},
  journal = {Nature Photonics},
  volume  = {9},
  number  = {9},
  pages   = {615--620},
  year    = {2015},
  doi     = {10.1038/nphoton.2015.153}
}

@article{huh_yung_vibronic_2017,
  title   = {Vibronic Boson Sampling: Generalized Gaussian Boson Sampling for Molecular Vibronic Spectra at Finite Temperature},
  author  = {Huh, Joonsuk and Yung, Man-Hong},
  journal = {Scientific Reports},
  volume  = {7},
  pages   = {7462},
  year    = {2017},
  doi     = {10.1038/s41598-017-07770-z}
}

@article{arrazola_dense_2018,
  title   = {Using Gaussian Boson Sampling to Find Dense Subgraphs},
  author  = {Arrazola, Juan Miguel and Bromley, Thomas R.},
  journal = {Physical Review Letters},
  volume  = {121},
  number  = {3},
  pages   = {030503},
  year    = {2018},
  doi     = {10.1103/PhysRevLett.121.030503}
}

@article{broome_photonic_2013,
  title   = {Photonic Boson Sampling in a Tunable Circuit},
  author  = {Broome, Matthew A. and Fedrizzi, Alessandro and Rahimi-Keshari, Saleh and Dove, Jonathan and Aaronson, Scott and Ralph, Timothy C. and White, Andrew G.},
  journal = {Science},
  volume  = {339},
  number  = {6121},
  pages   = {794--798},
  year    = {2013},
  doi     = {10.1126/science.1231440}
}

@article{tillmann_experimental_2013,
  title   = {Experimental Boson Sampling},
  author  = {Tillmann, Max and Daki{\'c}, Borivoje and Heilmann, René and Nolte, Stefan and Szameit, Alexander and Walther, Philip},
  journal = {Nature Photonics},
  volume  = {7},
  number  = {7},
  pages   = {540--544},
  year    = {2013},
  doi     = {10.1038/nphoton.2013.102}
}

@article{spagnolo2023non,
  title={Non-linear boson sampling},
  author={Spagnolo, Nicol{\`o} and Brod, Daniel J and Galv{\~a}o, Ernesto F and Sciarrino, Fabio},
  journal={npj Quantum Information},
  volume={9},
  number={1},
  pages={3},
  year={2023},
  doi = {10.1038/s41534-023-00676-x},
  publisher={Nature Publishing Group UK London}
}

@article{bouland_complexitytheoretic_2023,
  title   = {Complexity-theoretic foundations of {BosonSampling} with a linear number of modes},
  author  = {Bouland, Adam and Brod, Daniel and Datta, Ishaun
             and Fefferman, Bill and Grier, Daniel
             and Hern{\'a}ndez, Felipe and Oszmaniec, Micha{\l}},
  journal = {Physical Review X},
  volume  = {16},
  number  = {2},
  pages   = {021059},
  year    = {2026},
  doi     = {10.1103/xc7b-sjm5}
}

@inproceedings{bouland_exponential_2025,
  title = {Exponential Improvements to the Average-Case Hardness of {{BosonSampling}}},
  booktitle = {2025 {{IEEE}} 66th {{Annual Symposium}} on {{Foundations}} of {{Computer Science}} ({{FOCS}})},
  author = {Bouland, Adam and Datta, Ishaun and Fefferman, Bill and Hern{\'a}ndez, Felipe},
  year = 2025,
  month = dec,
  pages = {912--933},
  publisher = {IEEE},
  address = {Sydney, Australia},
  doi = {10.1109/FOCS63196.2025.00047},
  urldate = {2026-05-06},
  copyright = {https://doi.org/10.15223/policy-029},
  isbn = {979-8-3315-7132-0}
}

@misc{mhiri_boson_2026,
  title         = {Boson sampling beyond the dilute regime: Second moments and anti-concentration},
  author        = {Mhiri, Hela and Thomas, Hugo and Monbroussou, L{\'e}o
                   and Chabaud, Ulysse and Holmes, Zo{\"e} and Kashefi, Elham},
  year          = {2026},
  eprint        = {2604.14323},
  archiveprefix = {arXiv},
  primaryclass  = {quant-ph},
  doi           = {10.48550/arXiv.2604.14323},
  url           = {https://arxiv.org/abs/2604.14323}
}

@misc{kolarovszki_general_2026,
  title         = {General framework for anticoncentration and linear cross-entropy benchmarking in photonic quantum advantage experiments},
  author        = {Kolarovszki, Zolt{\'a}n and Kaposi, {\'A}goston
                   and Zimbor{\'a}s, Zolt{\'a}n and Oszmaniec, Micha{\l}},
  year          = {2026},
  eprint        = {2604.15258},
  archiveprefix = {arXiv},
  primaryclass  = {quant-ph},
  doi           = {10.48550/arXiv.2604.15258},
  url           = {https://arxiv.org/abs/2604.15258}
}

@article{loredo_boson_2017,
  title = {Boson {{Sampling}} with {{Single-Photon Fock States}} from a {{Bright Solid-State Source}}},
  author = {Loredo, J. C. and Broome, M. A. and Hilaire, P. and Gazzano, O. and Sagnes, I. and Lemaitre, A. and Almeida, M. P. and Senellart, P. and White, A. G.},
  year = 2017,
  month = mar,
  journal = {Physical Review Letters},
  volume = {118},
  number = {13},
  pages = {130503},
  publisher = {American Physical Society},
  doi = {10.1103/PhysRevLett.118.130503},
  urldate = {2025-07-01}
}

@article{wang_highefficiency_2017,
  title = {High-Efficiency Multiphoton Boson Sampling},
  author = {Wang, Hui and He, Yu and Li, Yu-Huai and Su, Zu-En and Li, Bo and Huang, He-Liang and Ding, Xing and Chen, Ming-Cheng and Liu, Chang and Qin, Jian and Li, Jin-Peng and He, Yu-Ming and Schneider, Christian and Kamp, Martin and Peng, Cheng-Zhi and H{\"o}fling, Sven and Lu, Chao-Yang and Pan, Jian-Wei},
  year = 2017,
  month = jun,
  journal = {Nature Photonics},
  volume = {11},
  number = {6},
  pages = {361--365},
  publisher = {Nature Publishing Group},
  issn = {1749-4893},
  doi = {10.1038/nphoton.2017.63},
  urldate = {2025-07-01},
  copyright = {2017 Springer Nature Limited},
  langid = {english}
}

@article{hoch_reconfigurable_2022,
  title   = {Reconfigurable continuously-coupled {3D} photonic circuit for {Boson Sampling} experiments},
  author  = {Hoch, Francesco and Piacentini, Simone and Giordani, Taira
             and Tian, Zhen-Nan and Iuliano, Mariagrazia and Esposito, Chiara
             and Camillini, Anita and Carvacho, Gonzalo and Ceccarelli, Francesco
             and Spagnolo, Nicol{\`o} and Crespi, Andrea and Sciarrino, Fabio
             and Osellame, Roberto},
  journal = {npj Quantum Information},
  volume  = {8},
  number  = {1},
  pages   = {55},
  year    = {2022},
  doi     = {10.1038/s41534-022-00568-6}
}

@article{young_atomic_2024,
  title = {An Atomic Boson Sampler},
  author = {Young, Aaron W. and Geller, Shawn and Eckner, William J. and Schine, Nathan and Glancy, Scott and Knill, Emanuel and Kaufman, Adam M.},
  year = 2024,
  month = may,
  journal = {Nature},
  volume = {629},
  number = {8011},
  pages = {311--316},
  issn = {0028-0836, 1476-4687},
  doi = {10.1038/s41586-024-07304-4},
  urldate = {2026-03-31},
  langid = {english}
}

@article{wang_boson_2019a,
  title = {Boson {{Sampling}} with 20 {{Input Photons}} and a 60-{{Mode Interferometer}} in a 1 0 14 -{{Dimensional Hilbert Space}}},
  author = {Wang, Hui and Qin, Jian and Ding, Xing and Chen, Ming-Cheng and Chen, Si and You, Xiang and He, Yu-Ming and Jiang, Xiao and You, L. and Wang, Z. and Schneider, C. and Renema, Jelmer J. and H{\"o}fling, Sven and Lu, Chao-Yang and Pan, Jian-Wei},
  year = 2019,
  month = dec,
  journal = {Physical Review Letters},
  volume = {123},
  number = {25},
  pages = {250503},
  issn = {0031-9007, 1079-7114},
  doi = {10.1103/PhysRevLett.123.250503},
  urldate = {2026-05-06},
  langid = {english}
}

@article{mezher_solving_2023,
  title = {Solving Graph Problems with Single Photons and Linear Optics},
  author = {Mezher, Rawad and Carvalho, Ana Filipa and Mansfield, Shane},
  year = 2023,
  month = sep,
  journal = {Physical Review A},
  volume = {108},
  number = {3},
  pages = {032405},
  issn = {2469-9926, 2469-9934},
  doi = {10.1103/PhysRevA.108.032405},
  urldate = {2026-05-06},
  langid = {english}
}

@article{shi_quantum_2022,
  title = {A Quantum Hash Function with Grouped Coarse-Grained Boson Sampling},
  author = {Shi, Jinjing and Lu, Yuhu and Feng, Yanyan and Huang, Duan and Lou, Xiaoping and Li, Qin and Shi, Ronghua},
  year = 2022,
  month = feb,
  journal = {Quantum Information Processing},
  volume = {21},
  number = {2},
  pages = {73},
  issn = {1570-0755, 1573-1332},
  doi = {10.1007/s11128-022-03416-w},
  urldate = {2026-05-06},
  langid = {english}
}

@article{yin_experimental_2025,
  title   = {Experimental quantum-enhanced kernel-based machine learning on a photonic processor},
  author  = {Yin, Zhenghao and Agresti, Iris and {de Felice}, Giovanni
             and Brown, Douglas and Toumi, Alexis and Pentangelo, Ciro
             and Piacentini, Simone and Crespi, Andrea and Ceccarelli, Francesco
             and Osellame, Roberto and Coecke, Bob and Walther, Philip},
  journal = {Nature Photonics},
  volume  = {19},
  number  = {9},
  pages   = {1020--1027},
  year    = {2025},
  doi     = {10.1038/s41566-025-01682-5}
}

@article{agresti_demonstration_2025,
  title = {Demonstration of Hardware Efficient Photonic Variational Quantum Algorithm},
  author = {Agresti, Iris and Paul, Koushik and Schiansky, Peter and Steiner, Simon and Yin, Zhenghao and Pentangelo, Ciro and Piacentini, Simone and Crespi, Andrea and Ban, Yue and Ceccarelli, Francesco and Osellame, Roberto and Chen, Xi and Walther, Philip},
  year = 2025,
  month = oct,
  journal = {Physical Review Research},
  volume = {7},
  number = {4},
  pages = {043021},
  issn = {2643-1564},
  doi = {10.1103/d7bb-ybfh},
  urldate = {2026-05-06},
  langid = {english}
}

@article{deng_solving_2023,
  title = {Solving {{Graph Problems Using Gaussian Boson Sampling}}},
  author = {Deng, Yu-Hao and Gong, Si-Qiu and Gu, Yi-Chao and Zhang, Zhi-Jiong and Liu, Hua-Liang and Su, Hao and Tang, Hao-Yang and Xu, Jia-Min and Jia, Meng-Hao and Chen, Ming-Cheng and Zhong, Han-Sen and Wang, Hui and Yan, Jiarong and Hu, Yi and Huang, Jia and Zhang, Wei-Jun and Li, Hao and Jiang, Xiao and You, Lixing and Wang, Zhen and Li, Li and Liu, Nai-Le and Lu, Chao-Yang and Pan, Jian-Wei},
  year = 2023,
  month = may,
  journal = {Physical Review Letters},
  volume = {130},
  number = {19},
  pages = {190601},
  issn = {0031-9007, 1079-7114},
  doi = {10.1103/PhysRevLett.130.190601},
  urldate = {2026-05-06},
  langid = {english}
}

@article{lim_classical_2025,
  title         = {Classical algorithms for estimating expectation values in linear optical circuits},
  author        = {Lim, Youngrong and Oh, Changhun},
  journal       = {Physical Review Letters},
  year          = {2026},
  note          = {Accepted for publication on 13 August 2026},
  doi           = {10.1103/bl9r-2gyb},
  eprint        = {2502.12882},
  archiveprefix = {arXiv},
  primaryclass  = {quant-ph}
}

@article{oh_quantuminspired_2024a,
  title = {Quantum-{{Inspired Classical Algorithm}} for {{Graph Problems}} by {{Gaussian Boson Sampling}}},
  author = {Oh, Changhun and Fefferman, Bill and Jiang, Liang and Quesada, Nicol{\'a}s},
  year = 2024,
  month = may,
  journal = {PRX Quantum},
  volume = {5},
  number = {2},
  pages = {020341},
  issn = {2691-3399},
  doi = {10.1103/PRXQuantum.5.020341},
  urldate = {2026-05-06},
  langid = {english}
}

@misc{thomas_shedding_2025,
  title         = {Shedding light on classical shadows: Learning photonic quantum states},
  author        = {Thomas, Hugo and Chabaud, Ulysse and Emeriau, Pierre-Emmanuel},
  year          = {2025},
  eprint        = {2510.07240},
  archiveprefix = {arXiv},
  primaryclass  = {quant-ph},
  doi           = {10.48550/arXiv.2510.07240},
  url           = {https://arxiv.org/abs/2510.07240}
}

@article{oh_classical_2021,
  title = {Classical Simulation of Lossy Boson Sampling Using Matrix Product Operators},
  author = {Oh, Changhun and Noh, Kyungjoo and Fefferman, Bill and Jiang, Liang},
  year = 2021,
  month = aug,
  journal = {Physical Review A},
  volume = {104},
  number = {2},
  pages = {022407},
  issn = {2469-9926, 2469-9934},
  doi = {10.1103/PhysRevA.104.022407},
  urldate = {2026-05-05},
  langid = {english}
}

@article{qi_regimes_2020,
  title = {Regimes of {{Classical Simulability}} for {{Noisy Gaussian Boson Sampling}}},
  author = {Qi, Haoyu and Brod, Daniel J. and Quesada, Nicol{\'a}s and {Garc{\'i}a-Patr{\'o}n}, Ra{\'u}l},
  year = 2020,
  month = mar,
  journal = {Physical Review Letters},
  volume = {124},
  number = {10},
  pages = {100502},
  issn = {0031-9007, 1079-7114},
  doi = {10.1103/PhysRevLett.124.100502},
  urldate = {2026-05-06},
  langid = {english}
}

@article{oh_classical_2024,
  title = {Classical Algorithm for Simulating Experimental {{Gaussian}} Boson Sampling},
  author = {Oh, Changhun and Liu, Minzhao and Alexeev, Yuri and Fefferman, Bill and Jiang, Liang},
  year = 2024,
  month = sep,
  journal = {Nature Physics},
  volume = {20},
  number = {9},
  pages = {1461--1468},
  issn = {1745-2473, 1745-2481},
  doi = {10.1038/s41567-024-02535-8},
  urldate = {2026-05-06},
  langid = {english}
}

@article{madsen_quantum_2022,
  title = {Quantum Computational Advantage with a Programmable Photonic Processor},
  author = {Madsen, Lars S. and Laudenbach, Fabian and Askarani, Mohsen Falamarzi. and Rortais, Fabien and Vincent, Trevor and Bulmer, Jacob F. F. and Miatto, Filippo M. and Neuhaus, Leonhard and Helt, Lukas G. and Collins, Matthew J. and Lita, Adriana E. and Gerrits, Thomas and Nam, Sae Woo and Vaidya, Varun D. and Menotti, Matteo and Dhand, Ish and Vernon, Zachary and Quesada, Nicol{\'a}s and Lavoie, Jonathan},
  year = 2022,
  month = jun,
  journal = {Nature},
  volume = {606},
  number = {7912},
  pages = {75--81},
  issn = {0028-0836, 1476-4687},
  doi = {10.1038/s41586-022-04725-x},
  urldate = {2026-05-07},
  langid = {english}
}

@article{he_timebinencoded_2017,
  title = {Time-{{Bin-Encoded Boson Sampling}} with a {{Single-Photon Device}}},
  author = {He, Yu and Ding, X. and Su, Z.-E. and Huang, H.-L. and Qin, J. and Wang, C. and Unsleber, S. and Chen, C. and Wang, H. and He, Y.-M. and Wang, X.-L. and Zhang, W.-J. and Chen, S.-J. and Schneider, C. and Kamp, M. and You, L.-X. and Wang, Z. and H{\"o}fling, S. and Lu, Chao-Yang and Pan, Jian-Wei},
  year = 2017,
  month = may,
  journal = {Physical Review Letters},
  volume = {118},
  number = {19},
  pages = {190501},
  issn = {0031-9007, 1079-7114},
  doi = {10.1103/PhysRevLett.118.190501},
  urldate = {2026-05-07},
  copyright = {http://link.aps.org/licenses/aps-default-license},
  langid = {english}
}

@article{carosini_programmable_2024,
  title = {Programmable Multiphoton Quantum Interference in a Single Spatial Mode},
  author = {Carosini, Lorenzo and Oddi, Virginia and Giorgino, Francesco and Hansen, Lena M. and Seron, Benoit and Piacentini, Simone and Guggemos, Tobias and Agresti, Iris and Loredo, Juan C. and Walther, Philip},
  year = 2024,
  month = apr,
  journal = {Science Advances},
  volume = {10},
  number = {16},
  pages = {eadj0993},
  issn = {2375-2548},
  doi = {10.1126/sciadv.adj0993},
  urldate = {2026-05-07},
  langid = {english}
}

@article{zhong_quantum_2020,
  title = {Quantum Computational Advantage Using Photons},
  author = {Zhong, Han-Sen and Wang, Hui and Deng, Yu-Hao and Chen, Ming-Cheng and Peng, Li-Chao and Luo, Yi-Han and Qin, Jian and Wu, Dian and Ding, Xing and Hu, Yi and Hu, Peng and Yang, Xiao-Yan and Zhang, Wei-Jun and Li, Hao and Li, Yuxuan and Jiang, Xiao and Gan, Lin and Yang, Guangwen and You, Lixing and Wang, Zhen and Li, Li and Liu, Nai-Le and Lu, Chao-Yang and Pan, Jian-Wei},
  year = 2020,
  month = dec,
  journal = {Science},
  volume = {370},
  number = {6523},
  pages = {1460--1463},
  issn = {0036-8075, 1095-9203},
  doi = {10.1126/science.abe8770},
  urldate = {2026-05-07},
  langid = {english}
}

@article{liu_robust_2025,
  title   = {{Gaussian} boson sampling with 1,024 squeezed states in 8,176 modes},
  author  = {Liu, Hua-Liang and Su, Hao and Deng, Yu-Hao and Gong, Si-Qiu
             and Gu, Yi-Chao and Tang, Hao-Yang and Jia, Meng-Hao and Wei, Qian
             and Song, Yu-Kun and Wang, Dong-Zhou and Zheng, Ming-Yang
             and Chen, Fa-Xi and Li, Li-Bo and Ren, Si-Yu and Zhu, Xue-Zhi
             and Wang, Mei-Hong and Chen, Yao-Jian and Liu, Yan-Fei
             and Song, Long-Sheng and Yang, Peng-Yu and Chen, Jun-Shi
             and An, Hong and Zhang, Lei and Gan, Lin and Yang, Guang-wen
             and Xu, Jia-Min and He, Yu-Ming and Wang, Hui and Zhong, Han-Sen
             and Chen, Ming-Cheng and Jiang, Xiao and Li, Li and Liu, Nai-Le
             and Su, Xiao-Long and Zhang, Qiang and Lu, Chao-Yang
             and Pan, Jian-Wei},
  journal = {Nature},
  volume  = {653},
  number  = {8115},
  pages   = {687--692},
  year    = {2026},
  doi     = {10.1038/s41586-026-10523-6}
}

@article{ehrenberg_transition_2025,
  title = {Transition of {{Anticoncentration}} in {{Gaussian Boson Sampling}}},
  author = {Ehrenberg, Adam and Iosue, Joseph T. and Deshpande, Abhinav and Hangleiter, Dominik and Gorshkov, Alexey V.},
  year = 2025,
  month = apr,
  journal = {Physical Review Letters},
  volume = {134},
  number = {14},
  pages = {140601},
  issn = {0031-9007, 1079-7114},
  doi = {10.1103/PhysRevLett.134.140601},
  urldate = {2026-05-07},
  langid = {english}
}

@article{eickmann_bridging_2026,
  title   = {Bridging chemistry and {Gaussian} boson sampling: A photonic hierarchy of approximations for molecular vibronic spectra},
  author  = {Eickmann, Jan-Lucas and Luo, Kai-Hong and Roiz, Mikhail
             and Lammers, Jonas and Atzeni, Simone and Pandey, Cheeranjiv
             and L{\"u}tkewitte, Florian and Shirazi, Reza G. and Schlue, Fabian
             and Brecht, Benjamin and Rybkin, Vladimir V. and Stefszky, Michael
             and Silberhorn, Christine},
  journal = {npj Quantum Information},
  volume  = {12},
  pages   = {89},
  year    = {2026},
  doi     = {10.1038/s41534-026-01250-x}
}

@article{wang_efficient_2020,
  title = {Efficient {{Multiphoton Sampling}} of {{Molecular Vibronic Spectra}} on a {{Superconducting Bosonic Processor}}},
  author = {Wang, Christopher S. and Curtis, Jacob C. and Lester, Brian J. and Zhang, Yaxing and Gao, Yvonne Y. and Freeze, Jessica and Batista, Victor S. and Vaccaro, Patrick H. and Chuang, Isaac L. and Frunzio, Luigi and Jiang, Liang and Girvin, S. M. and Schoelkopf, Robert J.},
  year = 2020,
  month = jun,
  journal = {Physical Review X},
  volume = {10},
  number = {2},
  pages = {021060},
  publisher = {American Physical Society},
  doi = {10.1103/PhysRevX.10.021060},
  urldate = {2026-05-11}
}

@article{zhu_largescale_2024,
  title = {Large-Scale Photonic Network with Squeezed Vacuum States for Molecular Vibronic Spectroscopy},
  author = {Zhu, Hui Hui and Sen Chen, Hao and Chen, Tian and Li, Yuan and Luo, Shao Bo and Karim, Muhammad Faeyz and Luo, Xian Shu and Gao, Feng and Li, Qiang and Cai, Hong and Chin, Lip Ket and Kwek, Leong Chuan and Nord{\'e}n, Bengt and Zhang, Xiang Dong and Liu, Ai Qun},
  year = 2024,
  month = jul,
  journal = {Nature Communications},
  volume = {15},
  number = {1},
  pages = {6057},
  issn = {2041-1723},
  doi = {10.1038/s41467-024-50060-2},
  urldate = {2026-05-13},
  langid = {english}
}

@article{cardin_photonnumber_2024,
  title = {Photon-Number Moments and Cumulants of {{Gaussian}} States},
  author = {Cardin, Yanic and Quesada, Nicol{\'a}s},
  year = 2024,
  month = nov,
  journal = {Quantum},
  volume = {8},
  pages = {1521},
  issn = {2521-327X},
  doi = {10.22331/q-2024-11-13-1521},
  urldate = {2026-05-13},
  langid = {english}
}

@article{lloyd1999quantum,
  title={Quantum computation over continuous variables},
  author={Lloyd, Seth and Braunstein, Samuel L},
  journal={Physical Review Letters},
  volume={82},
  number={8},
  pages={1784},
  year={1999},
  publisher={APS}
}

@article{fitzke_simulating_2023,
  title = {Simulating the Photon Statistics of Multimode {{Gaussian}} States by Automatic Differentiation of Generating Functions},
  author = {Fitzke, Erik and Niederschuh, Florian and Walther, Thomas},
  year = 2023,
  month = feb,
  journal = {APL Photonics},
  volume = {8},
  number = {2},
  pages = {026106},
  issn = {2378-0967},
  doi = {10.1063/5.0129638},
  urldate = {2026-05-13},
  langid = {english}
}

@misc{upreti_exponentiallyimproved_2026,
  title         = {Exponentially-improved effective descriptions of physical bosonic systems},
  author        = {Upreti, Varun and Quesada, Nicol{\'a}s and Chabaud, Ulysse},
  year          = {2026},
  eprint        = {2604.18720},
  archiveprefix = {arXiv},
  primaryclass  = {quant-ph},
  doi           = {10.48550/arXiv.2604.18720},
  url           = {https://arxiv.org/abs/2604.18720}
}

\appendix
\onecolumngrid

\section{General linear statistics}
\label{sec:general-weight}

We now discuss how linear statistics as presented in \cref{def:linear-statistics} can be generalized to weight-vectors with possibly negative and rational entries.

\paragraph{Negative-valued weight-vectors}
Assume that $\bomega \in \mathbb Z^m$, that is, $\bomega$ is a vector of signed integers. Let ${\bomega' = \bomega - \min_i \omega_i}\mathbf{1}$, where the sum is to be understood component-wise, so that $\bomega'$ is a vector of nonnegative integers. Then
\begin{equation}\label{eq:positive-shift-omega}
    \bomega' \cdot \s = \bomega \cdot \s - n \min_i \omega_i,
\end{equation}
giving
\begin{equation}
    \Pr_{\s}[\bomega \cdot \s = j] = \Pr_{\s}[\bomega' \cdot \s = j -  n \min_i \bomega_i], 
\end{equation}
and $j -  n \min_i \omega_i \geq 0$. We can thus consider the characteristic function $\chi_{\bomega'}$ and work in this shifted, non-negative-valued space. The characteristic function now reads 
\begin{align}
    \chi_{\bomega'}(k)
        & = \e[\s \sim D_U]{\exp\!\left(- \frac{2\pi i k}{N'}\, \big((\omega_1 - \min_i \omega_i) s_1 + \dots + (\omega_m -  \min_i \omega_i)s_m\big)\right)},
\end{align}
where $N' > n\left(\max_i \omega_i -\min_i\omega_i\right)$ and is even.
Physically (see \cref{thm:generatingFunction}), it corresponds to the expectation value of a layer of phase shifters where the phase shifter on the $r$-th mode has phase $2\pi i k(\omega_r -  \min_i \omega_i)/N'$.

\paragraph{Rational-valued weight-vectors}
We can pursue the idea of shifting weight vectors to consider rational-valued weight-vectors. Assume the vector $\bomega$ is now rational-valued, that is, $\omega_i = p_i / q_i$ and 
\begin{equation}
    p_i \in \mathbb{Z}_{\geq 0} \qquad q_i \in \mathbb{Z}_{>0}.
\end{equation}
Let $t = \text{lcm}(q_1, \cdots, q_m)$. Similar to \cref{eq:positive-shift-omega}, consider the weight vector $\bomega' = t\bomega$
yielding the shifted, integer-valued probability distribution
\begin{equation}
    \Pr_{\s}[\bomega \cdot \s = j] = \Pr_{\s}[\bomega' \cdot \s = jt], 
\end{equation}
where $tj$ is necessarily an integer. Now, the characteristic function reads
\begin{equation}
    \chi_{\bomega'}(k)
        = \e[\s \sim D_U]{\exp\!\left(- \frac{2\pi i k}{N'}\, \big(t\omega_1 s_1 + \dots + t\omega_m s_m\big)\right)},
\end{equation}
and $N' = tN$. Recall that the time complexity of the classical simulation algorithm is polynomial in (among other) $\log N' = \log(tN) = \log t + \log N$. Hence, the overhead is only logarithmic in $t$.

\section{Proof of Theorem 1}
\label{sec:montecarlo}

We prove the characteristic-function identity, the Fourier
estimator, and the importance-sampling bound, then combine
them with approximate characteristic-function evaluations
to establish \cref{cor:OWF}. As a reminder, $\Phi_m^n$ denotes the occupation vectors of $n$
photons in $m$ modes, and $f_{\bomega} (\s)=\bomega \cdot \s$.
For the specified input configuration and interferometer $U$,
$D_U$ denotes the output distribution.

\subsection{Characteristic function}\label{sec:CFApprox}
\bsgen*
\begin{proof}
Let
\begin{equation}
\widehat\Omega
=
\sum_{j=1}^m\omega_j\widehat n_j.
\end{equation}
Since
\begin{equation}
e^{\mathrm{i}\theta\widehat\Omega}|\s\rangle
=
e^{\mathrm{i}\theta\bomega\cdot\s}|\s\rangle,
\end{equation}
we obtain
\begin{align}
\chi_{\bomega}(\theta)
&=
\sum_{\s}
|\langle\s|\widehat U|\t\rangle|^2
e^{\mathrm{i}\theta\bomega\cdot\s} \\
&=
\langle\t|
\widehat U^\dagger
e^{\mathrm{i}\theta\widehat\Omega}
\widehat U
|\t\rangle \\
&=
\langle\t|
\widehat V_{\Omega,\theta}
|\t\rangle.
\end{align}
\end{proof}

This characteristic function may be approximated, for instance, using the algorithm in the proof of Theorem~2 of \cite{lim_classical_2025}, which gives an algorithm for estimating transition amplitudes of product states in Boson Sampling that scale polynomially in $m$, $\varepsilon^{-1}$, and $\log(1/\delta)$. In the next proofs, we use the shorthand
$\chi_{\bomega}(k)=\chi_{\bomega}(-2\pi k/N)$
for $k\in\{0,\ldots,N-1\}$.

\subsection{Fourier estimator of the cumulative distribution}\label{sec:EstimateS}

\mcfourier*
\begin{proof}

Since $f_{\bomega}(\bm s)$ is integer-valued, the function
\begin{equation}\label{eq:characteristicFunction}
\chi_{\bomega}(k)
\;=\;
\e[\s \sim D_U]{\exp\!\left(- \frac{2\pi i k}{N}\, f_{\bomega}(\bm s) \right)},
\end{equation}
is the discrete Fourier transform of the probability mass function of the linear statistic modulo $N$. 
Indeed,
\begin{equation}
\label{eq:arkhipov-char}
\chi_{\bomega}(k) =
     \sum_{j=0}^{N-1} \Pr[\bm \omega \cdot \bm s = j \mod N]\,
\exp\!\left(-\frac{2\pi i k}{N} j\right),
\end{equation}
Because $0 < f_{\bomega}(\s) < N$, the individual probabilities
can be recovered by the inverse discrete Fourier transform:
\begin{equation}\label{eq:inverse-dft}
    \Pr[\bm \omega \cdot \bm s = j]
        = \frac{1}{N} \sum_{k=0}^{N-1} \chi_{\bomega}(k)\, \exp\!\left(\frac{2\pi i k}{N} j\right).
\end{equation}
Using this formula, the cumulative distribution admits the following Fourier representation:
\begin{align}
    S(x) 
        & = \sum_{j=0}^{x} \Pr[\bm \omega \cdot \bm s = j] \\
        & = \sum_{j=0}^{x} \frac{1}{N} \sum_{k=0}^{N-1} \chi_{\bomega}(k)\, \exp\!\left(\frac{2\pi i k}{N} j\right)\\
        & =  \frac{1}{N}\sum_{k=0}^{N-1} \chi_{\bomega}(k)\, G_N(k;x) \\ 
        & =  \sum_{k=0}^{N-1} q(k) \lp\frac{\chi_{\bomega}(k)\, G_N(k;x)}{N q(k)}\rp\\ 
        & = \e[k\sim q]{\frac{\chi_{\bomega}(k)\, G_N(k;x)}{N q(k)}}\label{eq:ZSx0},
\end{align}
where $q$ is a probability distribution with full support on $\{0, \ldots, N-1\}$ and 
\begin{equation}\label{eq:GN}
    G_N(k;x) = \sum_{j=0}^{x} \exp\!\left(\frac{2\pi i k}{N} j\right)
    = \begin{cases}
        x + 1, 
            & k = 0, \\[6pt]
        \displaystyle \frac{1 - \exp\!\left(\frac{2\pi i k}{N}(x+1)\right)} {1 - \exp\!\left(\frac{2\pi i k}{N}\right)},
            & k \neq 0.
\end{cases} 
\end{equation}
Provided $q(k)$ can be sampled from and
\begin{equation}\label{app:eq:boundedModulus}
    \left| \frac{\chi_{\bomega}(k)\, G_N(k;x)}{N q(k)} \right| = O(R),
\end{equation}
Since $S(x)$ is real, we apply Hoeffding's inequality to the
real part of the random variable in \cref{eq:ZSx0}, which has
expectation $S(x)$ and satisfies the same range bound. All
Monte Carlo averages below are understood to use this real part. Thus, Hoeffding's inequality \cite{hoeffding_class_1948} implies that $S(x)$ can be estimated up to additive error $\varepsilon$ with failure
probability $\delta$ using $\poly{\varepsilon^{-1}, R, \log\delta^{-1}}$ independent samples from
$q(k)$.
\end{proof}

\subsection{Importance sampling and proof of Theorem 1}\label{sec:thm1}

\cref{thm:estimateS} shows that the sample complexity of estimating $S(x)$ depends on the range $R$ of the estimator. We now introduce a distribution for which $R = O(\log N)$ and the expected sampling time is $\poly{\log N}$.

\importance*

\begin{proof}
We describe how to sample from $q$ with inverse-transform sampling and a
rejection step. First, we must choose which of the three cases of (\cref{eq:tk})
to sample from by sampling according to their respective masses. If the first case
is picked, return $k=0$. Otherwise, we sample from either of the two remaining
cases using inverse-transform sampling and a rejection step. If the last branch was selected, we use the change of variables $k'=N-k$.
Both nonzero branches can be sampled using the procedure below
on $\{1,\ldots,N/2\}$; for the first nonzero branch, reject and resample
whenever the generated index equals $N/2$. 

In order to sample an integer $k$ with probability proportional to $1/k$, we use
as a proxy the continuous distribution with PDF
\begin{equation}
    P_Y(k) \propto \int_{k}^{k+1} \frac{1}{x} dx = \log(k+1) - \log(k) = \log\left(1 + \frac{1}{k}\right),
\end{equation}
and round down to the nearest integer. Its CDF is 
\begin{equation}
    F_Y(k) = \frac{\log(k+1)}{\log(N/2 + 1)},
\end{equation}
thus, via inverse-transform sampling, we sample from this distribution by
sampling $v \sim \text{Unif[0, 1]}$ and outputting $k = \lfloor (N/2 +1)^v \rfloor $. Sampling from the continuous proxy generates integers $k$ with
probability proportional to $\log(1 + 1/k)$. To sample from our
target distribution, we perform a rejection step. To this end, we must find a
constant $\mathcal C$ such that 
\begin{equation}
    \frac{1}{k} \leq \mathcal C \log(1+1/k) \Leftrightarrow \frac{1}{\mathcal C} \leq k\log(1+1/k).
\end{equation}
Let $f(x) = x\log(1+1/x)$. As $\frac{d}{dx}f(x) = \log(1+1/x) - \frac{1}{x+1} >
0$ and $\lim_{x \to \infty} f(x) = 1$, its minimum over the integers is reached
at $x=1$, giving $\mathcal C = \frac{1}{\log 2}$. Therefore, a
candidate integer $k$ is accepted if a uniformly drawn $u \sim\text{Unif}[0, 1]$
satisfies 
\begin{equation}
    u \leq \frac{\log 2}{k \log (1+1/k)}.
\end{equation}
The acceptance probability converges quickly to $\log(2) \approx 0.7$, therefore
a constant number of attempts are required to sample from the true distribution.
Finally, we note that the normalization factor $\mathcal N$ satisfies 
\begin{align}
    \mathcal N 
        & =  1 + H_{\lfloor N/2\rfloor - 1} +  H_{\lfloor N/2\rfloor} \\
        &  = 1 + 2\gamma + \ln \left( \lfloor N/2\rfloor - 1 \right) + \ln\lfloor N/2\rfloor  + O(1/N),
\end{align}
where $H_n = \sum_{k = 1}^n \frac{1}{k}$ is the $n$-th harmonic number and
$\gamma \approx 0.577$ is the Euler-Mascheroni constant. Since each proposal can be generated 
and tested in time polynomial in $\log N$, a sample from $q$ can be generated in expected time $\poly{\log N}$.

It remains to prove the second claim of the lemma. For $k = 0$, using $G_N(0;x) = x+1$ and 
$q(0) = 1/\mathcal{N}$, we have 
\begin{equation}
    \frac{|G_N(0;x)|}{Nq(0)}
    =
    \frac{(x+1)\mathcal{N}}{N}
    \leq \mathcal{N}.
\end{equation}
For $1 \leq k \leq N-1$,
\begin{equation}
    |G_N(k,x)|
    = 
    \frac{\left|\sin\!\left(\pi k(x+1)/N\right)\right|
    }{
    \left|\sin\!\left(\pi k/N\right)\right|
    }.
\end{equation}
Using
\begin{equation}
    \left|\sin\!\left(\frac{\pi k}{N}\right)\right|
    \geq
    \frac{2\min\{k,N-k\}}{N},
\end{equation}
we obtain
\begin{equation}
    |G_N(k;x)|
    \leq
    \frac{N}{2\min\{k,N-k\}}.
\end{equation}
Since
\begin{equation}
    t(k) = \frac{1}{\min\{k,N-k\}},
\end{equation}
for $k \neq 0$, it follows that
\begin{equation}
   \frac{|G_N(k;x)|}{Nq(k)}
   =
   \frac{\mathcal N |G_N(k;x)|}{Nt(k)}
   \leq
   \frac{\mathcal{N}}{2}.
\end{equation}
Finally, since $\mathcal{N} = O(\log N)$, we have
\begin{equation}\label{eq:strongerbound}
    \left|\frac{G_N(k;x)}{Nq(k)}\right|
    = O(\log N).
\end{equation}
uniformly in $x$ and $k$. Therefore using $| \chi_{\bomega}(k)| \leq 1$, the 
condition in \cref{eq:boundedModulus} holds with $R = O(\log N)$.
\end{proof}

We now combine the preceding results to prove \cref{cor:OWF}, taking into account that the characteristic function
can only be evaluated approximately.

\owf*

\begin{proof}
    By \cref{thm:estimateS}, define
    \begin{equation}
        Z(k;x) 
        =
        \frac{\chi_{\bomega}(k)G_N(k;x)}{Nq(k)},
    \end{equation}
    so that
    \begin{equation}
        S(x) =  \e[k \sim q] {Z(k;x)}.
    \end{equation}
        For $M$ independent samples $k_1, \ldots, k_M \sim q$, let 
    \begin{equation}
        {\widehat S}_M(x)  
        =
        \frac{1}{M} \sum_{i = 1}^M Z(k_i, x).
    \end{equation}
    By \cref{lm:Q}, there exists a constant $C>0$ such that
    \begin{equation}
        |Z(k;x)| \leq C\log N,
    \end{equation}
    uniformly in $x$ and $k$. Hence, Hoeffding's inequality implies that
    \begin{equation}\label{eq:varepsilonMC}
        |\widehat S_M(x) - S(x) |\leq \varepsilon_{\mathrm{MC}},
    \end{equation}
    with a probability at least $1-\delta_{\mathrm{MC}}$, provided
    \begin{equation}\label{eq:sample-number}
        M 
        = 
        O\!\left(\frac{\log^2 N}{\varepsilon_{\mathrm{MC}}^2}
        \log\frac{1}{\delta_{\mathrm{MC}}}\right).
    \end{equation}
    The characteristic function $\chi_{\bomega}(k_i)$ is aproximated
    using the algorithm for estimating transition-amplitudes in \cite{lim_classical_2025}. Let $\widetilde\chi_{\bomega}(k_i)$ be an estimate satisfying
    \begin{equation}
        |\widetilde\chi_{\bomega}(k_i) - \chi_{\bomega}(k_i) |\leq \varepsilon_{\chi},
    \end{equation}
    with probability at least $1-\eta_\chi$ and define
    \begin{equation}
        \widetilde Z(k_i;x)
        =
        \frac{\widetilde\chi_{\bomega}(k_i)G_N(k_i;x)}{Nq(k_i)}.
    \end{equation}
    By \cref{eq:strongerbound},
    \begin{equation}
        \begin{aligned}
        \left| \widetilde Z(k_i,x)-Z(k_i;x)\right|
        =
        \left| \widetilde\chi_{\bomega}(k_i)
        -
        \chi_{\bomega}(k_i)
        \right|
        \left|
        \frac{G_N(k_i;x)}{Nq(k_i)}
        \right| \leq C\varepsilon_{\chi}\log N.
        \end{aligned}
    \end{equation}
    Therefore, choosing 
    \begin{equation}
        \varepsilon_{\chi} 
        = 
        \frac{\varepsilon}{2C\log N}
    \end{equation}
    ensures that
    \begin{equation}
        \left| \widetilde Z(k_i,x)-Z(k_i;x)\right|
        \leq \frac{\varepsilon}{2}.
    \end{equation}
    Let
    \begin{equation}
        \overline S_M(x)
        =
        \frac{1}{M} \sum_{i = 1}^M \widetilde Z(k_i, x).
    \end{equation}
    If all $M$ characteristic function estimates succeed, then
    \begin{equation}
        \left| \widehat S_M(x)-\overline S_M(x)\right|
        \leq
        \frac{1}{M} \sum_{i = 1}^M \left| Z(k_i, x) -\widetilde Z(k_i, x)\right| \leq
        \frac{\varepsilon}{2}.
    \end{equation}
    By the union bound, all  $M$ estimates satisfy the above inequality with probability at least $1-M\eta_\chi$. When this event and the Monte Carlo bound in \cref{eq:varepsilonMC} both hold, the triangle inequality
    gives
    \begin{equation}
        \begin{aligned}
         \left| S(x) - \overline S_M(x)\right|\leq
         \left| S(x)-\widehat S_M(x)\right|
         +
         \left| \widehat S_M(x)-\overline S_M(x)\right| \leq
         \varepsilon_{\mathrm{MC}}+\frac{\varepsilon}{2}.
        \end{aligned} 
    \end{equation}
    Taking $\varepsilon_{\mathrm{MC}} =\frac{\varepsilon}{2}$ and $M\eta_\chi = \delta_{\mathrm{MC}} = \frac{\delta}{2}$,
    we obtain
    \begin{equation}
        \Pr[\left| S(x) - \overline S_M(x)\right| \leq \varepsilon] \geq 1-\delta.
    \end{equation}
    The required precision for each characteristic function evaluation is
    \begin{equation}
        \varepsilon_\chi = \Theta\!\left(\frac{\varepsilon}{\log N}\right),
    \end{equation}
    and its failure probability is
    \begin{equation}
        \eta_\chi =\Theta\!\left(\frac{\delta}{M}\right).
    \end{equation}
   With the choices above, \cref{eq:sample-number}
   gives
   \begin{equation}\label{eq:SampComp}
        M
        =
        O\!\left(
        \frac{\log^2 N}{\varepsilon^2}
        \log\frac{1}{\delta}
        \right).
   \end{equation}
   By \cref{lm:Q}, each sample $k_i\sim
   q$ can be generated in 
   expected time $\poly{\log N}$. Moreover, the algorithm in the proof of Theorem~2 of \cite{lim_classical_2025} approximates each characteristic-function value 
   in time polynomial in $m$, $\varepsilon_\chi^{-1}$, and $\log(1/\eta_\chi)$. Hence, when we combine these running times with the sample complexity in \cref{eq:SampComp}, we get the total expected running time
   \begin{equation}
       T_{\text{total}} = O\left(\frac{m^2 \log^2 N}{\epsilon^4} \log(1/\delta)\log\left(\frac{\log^2 N}{\epsilon^2\delta}\log(1/\delta)\right)\right).
   \end{equation} This proves \cref{cor:OWF}. The nested Monte Carlo procedure leads to the $\varepsilon^{-4}$ dependence, while the additional $\log(M/\delta)$ factor comes from controlling the $M$ independent inner estimators. A joint estimator could remove the latter factor and may also improve the dependence on $\epsilon$. Moreover, whenever
    $\log N=\poly{m}$, the algorithm has expected polynomial
    running time in $m$, $\varepsilon^{-1}$, and
    $\log(1/\delta)$. 
\end{proof}

\section{Mode grouping for the validation of Boson Sampling}
\label{sec:proofModes}

In this section we describe how mode-grouping can be expressed in the framework
of linear statistics of Boson Sampling.
The mode-grouping strategy, as introduced in \cite{seron_efficient_2024}, is
defined as follows. The $m$ output modes are partitioned into $K$ disjoint
subsets $\Kc_1, \cdots, \Kc_K$, which, up to permutation, are contiguous. We
define the group-number operator as $\bm{\hat n}_{k} = \sum_{j \in \Kc_K} \hat
n_j$ and $\bm {\hat N} = (\bm{\hat n}_{1}, \cdots, \bm{\hat n}_{K})$. Using the
notation of \cref{eq:phimn}, group occupations are described by elements of
$\Phi_K^n$. For a mode-group occupation $\bm\ell \in \Phi_K^n$, indicating that
$\ell_i$ photons occupy the $i$-th group $\Kc_i$, we denote by
$\mathscr{O}_{\bm\ell}\subseteq \Phi_m^n$ the set of photon occupations whose
mode-grouping is $\bm \ell$, that is,
\begin{equation}
   \mathscr{O}_{\bm\ell} = \left\{\s \ | \ \s \in \Phi_m^n,\, \textstyle \sum_{j \in \Kc_i} s_j = \ell_i, \forall\ 1 \leq i \leq K\right\}.
\end{equation}

For a
non-negative integer $k$, the characteristic function associated with a
mode-grouping weight-vector $\bomega^{(\textsc{mg})}$ reads
\begin{align}
    \chi_{\bomega^{(\textsc{mg})}}(k) 
        & = \e[\bm{s} \sim D_U]{\exp(-\frac{2 \pi \imath k \bm \omega^{(\textsc{mg})} \cdot \bm s}{N})}\\
        & = \sum_{\bm s \in \Phi_m^n} \Pr[\bm s] \exp(-\frac{2 \pi \imath k \bm \omega^{(\textsc{mg})} \cdot \bm s}{N}) \\
        & = \sum_{\bm \ell \in \Phi_K^n} \exp(-\frac{2 \pi \imath k \bm \eta \cdot \bm \ell}{N}) \sum_{\bm s \in \mathscr{O}_{\bm \ell}} \Pr[\bm s] \\
        & = \sum_{\bm \ell \in \Phi_K^n} \Pr[\bm \ell] \exp(-\frac{2 \pi \imath k \bm \eta \cdot \bm \ell}{N}) \label{eq:ChiFourier}
\end{align}

where we defined

\begin{align}
        \bm\omega^{(\textsc{mg})} 
            & = (\underbrace{1, \cdots, 1}_{|\mathcal{K}_1|\ \text{times}}, \underbrace{n+1, \cdots, n+1}_{|\mathcal{K}_2|\ \text{times}}, \cdots, \underbrace{(n+1)^{K-1}, \cdots, (n+1)^{K-1}}_{|\mathcal{K}_K|\ \text{times}}) \label{eq:restrictedOmega},\\
        \bm \eta 
            & = (1, n+1, (n+1)^2, \cdots, (n+1)^{K-1}) \label{eq:defEta}.
\end{align}

The probability of observing a group-outcome $\bm \ell$ is obtained by taking
the inverse Fourier transform of the characteristic function as defined in
\cref{eq:ChiFourier}. This gives
\begin{equation}\label{eq:mode-binning-exp}
    \begin{aligned}
        \Pr_{\bm \ell}[\bm \eta \cdot \bm \ell = k]
        = \frac{1}{N} \sum_{j=0}^{N-1} \exp\!\left(\frac{2 \imath \pi j k}{N}\right)\chi_{\bomega^{(\textsc{mg})}}(j)= \e[j \sim \text{Unif}(0, \cdots, N-1)]{\exp\!\left(\frac{2 \imath \pi j k}{N}\right)\chi_{\bomega^{(\textsc{mg})}}(j)},
    \end{aligned}
\end{equation}
where we chose $q$ to be the uniform distribution. Let $\hat\chi_{\bomega^{(\textsc{mg})}}(j)$ be an estimate obtained using the algorithm in the proof of Theorem~2 of \cite{lim_classical_2025} satisfying  $|\hat\chi_{\bomega^{(\textsc{mg})}}(j)-\chi_{\bomega^{(\textsc{mg})}}(j)| \leq \varepsilon_{est}$ , then
\begin{equation}
    \left|\exp\!\left(\frac{2 \imath \pi j k}{N}\right)\hat\chi_{\bomega^{(\textsc{mg})}}(j)\right| \leq 1 + \varepsilon_{est},
\end{equation}
since it is a product of terms whose modulus is at most $1$ and an additive factor $\varepsilon_{est}$. Thus, Monte-Carlo sampling over \cref{eq:mode-binning-exp} with $M = O((1+\varepsilon_{est})^2/\varepsilon_{MC}^2)$ gives an estimate of the group probability $\Pr_{\bm
\ell}[\bm \eta \cdot \bm \ell = k]$ within additive error at most $\epsilon = \varepsilon_{MC} + \varepsilon_{est}$ with negligible failure probability. Each of the $M$ estimators takes time $O(m^2 /\varepsilon_{est}^2)$
with the classical algorithm of \cite{lim_classical_2025}, giving an overall time complexity of $O(m^2 /\epsilon^4)$.

Formally, we obtain \cref{cor:modeBinning} as
formulated in the main text.

\modebins*

\end{document}